\documentclass[journal]{IEEEtran}
\twocolumn \def\twocolumnmode{} 

\usepackage{cite} 
\usepackage{amsmath}%
\usepackage{amsfonts}%
\usepackage{amssymb}%
\usepackage{graphicx}
\newtheorem{theorem}{Theorem}

\newtheorem{lemma}[theorem]{Lemma}

\newtheorem{remark}[theorem]{Remark}

\newenvironment{proof}[1][Proof]{\textbf{#1.} }{\ \rule{0.5em}{0.5em}}

\usepackage{epstopdf}

\usepackage{pdflscape}
\usepackage{multirow}

\usepackage{pgf}
\usepackage{tikz}
\usetikzlibrary{arrows,automata}
\usepackage[latin1]{inputenc}
\usepackage{verbatim}

\newcommand{\Ts}{T_s}  
 
\newcommand{\nsymb}{{b}}		
\newcommand{\nsamp}{{n}}		

\newcommand{\X}{\textbf{X}}
\newcommand{\Y}{\textbf{Y}}

\newcommand{\x}{\textbf{x}}
\newcommand{\y}{\textbf{y}}

\newcommand{\SNR}{\textsf{SNR}}
\newcommand{\Xsymbol}{{X}}
\newcommand{\xsymbol}{{x}}
\newcommand{\Xsymb}[1]{{\Xsymbol}_{#1}}
\newcommand{\xsymb}[1]{{\xsymbol}_{#1}}

\newcommand{\R}{{R}}
\newcommand{\erfc}{\text{erfc}}
\newcommand{\erf}{\text{erf}}

\usepackage{scalefnt}		

\newcommand{\Psiv}{\boldsymbol{\Psi}}
\newcommand{\Xsamp}{\textsf{X}}

\renewcommand{\H}{\textsf{H}} 

\let\labelindent\relax
\usepackage{enumitem}

\ifCLASSINFOpdf
\else
\fi
\begin{document}
%
\title{Models and Information Rates for \\Wiener Phase Noise Channels}
%
%
%

\author{Hassan~Ghozlan,~\IEEEmembership{Member,~IEEE},
        and~Gerhard~Kramer,~\IEEEmembership{Fellow,~IEEE}
\thanks{
The work of H. Ghozlan was supported by a USC Annenberg Fellowship and 
the National Science Foundation (NSF) under Grant CCF-09-05235.
The work of G. Kramer was supported by an Alexander von Humboldt Professorship endowed by
the German Federal Ministry of Education and Research, as well as by the NSF under Grant CCF-09-05235.
Part of the material in this paper was presented
at the IEEE International Symposium on Information Theory, Istanbul, Turkey, July, 2013,
at the Global Communications Conference, Atlanta, GA, December, 2013 and
at the IEEE International Symposium on Information Theory, Honolulu, HI, June/July, 2014.}%
\thanks{H. Ghozlan was with the Department
of Electrical Engineering, University of Southern California, Los Angeles,
CA 90089, USA.
He is now with Intel Corporation, Hillsboro, OR 97124, USA.
(e-mail: hassan.ghozlan@intel.com).}
\thanks{G. Kramer is with the Institute for Communications Engineering, 
Technical University of Munich, 80333 Munich, Germany (email: gerhard.kramer@tum.de).}
\thanks{Copyright (c) 2014 IEEE. Personal use of this material is permitted. 
However, permission to use this material for any other purposes must be obtained from the IEEE 
by sending a request to pubs-permissions@ieee.org.}%
}

%
%

\markboth{Journal,~Vol.~, No.~, Month~Year}%
{Shell \MakeLowercase{\textit{et al.}}: Bare Demo of IEEEtran.cls for Journals}
%



\maketitle

%
\IEEEpeerreviewmaketitle

\newcommand{\codered}[1]{{\color[rgb]{0,0,0} {#1}}}
\newcommand{\TexFigPath}{./fig-tikz-new-wmestimate/}			

\begin{abstract}
A waveform channel is considered where the transmitted signal is corrupted by Wiener phase noise
and additive white Gaussian noise.
A discrete-time channel model that takes into account the effect of filtering on the phase noise is developed.
The model is based on a multi-sample receiver, i.e.,
an integrate-and-dump filter whose output is sampled at a rate higher than the signaling rate.
It is shown that, at high Signal-to-Noise Ratio (SNR), 
the multi-sample receiver achieves a rate that grows logarithmically with the SNR
if the number of samples per symbol (oversampling factor) grows with the cubic root of the SNR.
Moreover, the pre-log factor is at least 1/2 and can be achieved by amplitude modulation.
For an approximate discrete-time model of the multi-sample receiver,
the capacity pre-log at high SNR is shown to be at least 3/4
if the number of samples per symbol grows with the square root of the SNR. 
The analysis shows that phase modulation achieves a pre-log of at least 1/4 
while amplitude modulation still achieves a pre-log of 1/2.
This is strictly greater than the capacity pre-log of the (approximate) discrete-time Wiener phase noise channel with only one sample per symbol, 
which is 1/2.
Numerical simulations are used to compute lower bounds on the information rates achieved by the multi-sample receiver.
The simulations show that oversampling is beneficial 
for both strong and weak phase noise at high SNR.
In fact, the information rates are sometimes substantially larger than when using commonly-used approximate discrete-time models.
\end{abstract}


\section{Introduction}
Phase noise arises due to the instability of oscillators \cite{Demir2000} in communication systems
such as satellite communication \cite{CasiniDVBS2PN2004}, 
microwave links \cite{DurisiICC2013} or optical fiber communication \cite{TkachDFB1986}.
The statistical characterization of the phase noise depends on the application.
In systems with phase-locked loops (PLL), 
the residual phase noise follows a Tikhonov distribution \cite{ViterbiPLL1963}.
In Digital Video Broadcasting DVB-S2, an example of a satellite communication system, 
the phase noise process is modeled by the sum of the outputs of two infinite-impulse response filters 
driven by the same white Gaussian noise process \cite{CasiniDVBS2PN2004}.
In optical fiber communication, the phase noise in laser oscillators 
is modeled by a Wiener process \cite{TkachDFB1986}.

A commonly studied \emph{discrete-time} channel model is
\begin{align}
 Y_k = \Xsymb{k} ~ e^{j \Theta_k} + Z_k
\label{eq:dt-pn-ch}
\end{align}
where $\{Y_k\}$ are the output symbols, $\{\Xsymb{k}\}$ are the input symbols, $\{\Theta_k\}$ is the phase noise process and 
$\{Z_k\}$ is additive white Gaussian noise (AWGN).
For example,
Katz and Shamai \cite{Katz2004} studied the model (\ref{eq:dt-pn-ch}) when $\{\Theta_k\}$ is 
independent and identically distributed (i.i.d.) according to $p_{\Theta}(\cdot)$,
when $\Theta$ is uniformly distributed (called a non-coherent AWGN  channel)
and 
when $\Theta$ has a Tikhonov (or von Mises) distribution (called a partially-coherent AWGN channel).
The i.i.d. Tikhonov phase noise models the residual phase error in systems 
with phase-tracking devices, e.g., phase-locked loops (PLL) and ideal interleavers/deinterleavers.
Katz and Shamai \cite{Katz2004} characterized some properties of the capacity-achieving distribution.
Tight lower bounds on the capacities of memoryless non-coherent and partially coherent AWGN channels 
were computed by solving an optimization problem numerically in \cite{Katz2004} and \cite{Hou2003PCAWGN}, respectively.
Lapidoth studied in \cite{LapidothPhaseNoise2002} a phase noise channel (\ref{eq:dt-pn-ch}) at high SNR. 
He considered both memoryless phase noise and phase noise with memory. 
He showed that the capacity grows \emph{logarithmically} with the SNR with a pre-log factor 1/2
if the differential entropy rate of $\{\Theta_k\}$ is finite.
The pre-log 1/2 is achieved by amplitude modulation only, and
the phase modulation contributes only a bounded number of bits.
Dauwels and Loeliger \cite{Dauwels2008} proposed
a particle filtering method to compute information rates for discrete-time continuous-state channels with memory
and applied the method to (\ref{eq:dt-pn-ch}) for Wiener phase noise and auto-regressive-moving-average (ARMA) phase noise.
Barletta, Magarini and Spalvieri \cite{BarlettaLB}
computed lower bounds on information rates for (\ref{eq:dt-pn-ch}) with Wiener phase noise
by using the auxiliary channel technique proposed in \cite{Arnold2006}
and they computed upper bounds in \cite{BarlettaUB}.
They also developed a lower bound based on Kalman filtering in \cite{Barletta2012KalmanLB}.
Barbieri and Colavolpe \cite{Barbieri2011}
computed lower bounds with an auxiliary channel slightly different from \cite{BarlettaLB}.
Durisi et al. \cite{DurisiITA2013,DurisiCommTrans2014} developed capacity bounds for a multi-input multi-output (MIMO) extension of 
the discrete-time model (\ref{eq:dt-pn-ch}).

A less-studied model since 1990 is the continuous-time model for phase noise channels,
also called a \emph{waveform} phase noise channel, see Fig. \ref{fig:waveform_basic}.
The received waveform $r(t)$ is
\begin{align}
  r(t) = x(t) \ e^{j \theta(t)} + n(t),
  \text{ for } t \in \mathbb{R}
  \label{eq:waveform_ch_intro}
\end{align}
where $x(t)$ is the transmitted waveform
while $n(t)$ and $\theta(t)$ are additive noise and phase noise, respectively.
\begin{figure}
   \centering
	\resizebox{ \ifdefined\twocolumnmode 0.85\columnwidth \else 0.6\textwidth \fi }{!}{{\scalefont{0.85}
\begin{tikzpicture}[->,>=stealth', shorten >=1pt, auto, node distance=2.25cm, semithick]
  \tikzstyle{circ}=[circle,thick,draw=blue!75,minimum size=3pt]
  \tikzstyle{rect} = [draw=blue, rectangle, inner sep=10pt, inner ysep=10pt]
	\tikzset{text/.default=}
  
  \node[rect] (TX) [] {Transmitter};
  \node[circ] (PN) [right of=TX,xshift=0.5cm] {$\times$};
  \node[circ] (AN) [right of=PN, xshift=-1cm] {$+$};
  \node[rect] (RX) [right of=AN] {Receiver};
  \node[text] (N_t) [below of=AN, yshift=+0.75cm] {$n(t)$};
  \node[text] (Theta_t) [below of=PN, yshift=+0.75cm] {$e^{j \theta(t)}$};

  \path 
  (TX) edge node {$x(t)$} (PN)
  (PN) edge node {} (AN)
  (AN) edge node {$r(t)$} (RX)
  
  (Theta_t) edge node {} (PN)  
  (N_t) edge node {} (AN)
;

\draw [help lines];
\end{tikzpicture}
}}
   \caption{Waveform phase noise channel.}
	 \label{fig:waveform_basic}
\end{figure}
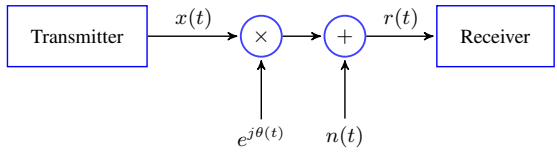
Continuous-time \emph{white} phase noise was recently considered in \cite[Section IV.C]{Goebel2011} and \cite{LucaWhiteISIT2014, LucaWhiteCROWN2014}.
We study the waveform channel (\ref{eq:waveform_ch_intro}) with \emph{Wiener} phase noise.
A detailed description of the channel is given in Sec. \ref{sec:ct-model}.
This model is reasonable, for example, for optical fiber communication with low to intermediate power and laser phase noise,
see \cite{Foschini1988IT, Foschini1988Comm, FoschiniCom1989}.
Since the sampling of a continuous-time Wiener process yields a discrete-time Wiener process (Gaussian random walk),
it is tempting to use the model (\ref{eq:dt-pn-ch})
with $\{\Theta\}$ as a discrete-time Wiener process.
However, this ignores the effect of \emph{filtering} prior to sampling.
It was pointed out in \cite{Foschini1988IT} that 
``even coherent systems relying on amplitude modulation 
(phase noise is obviously a problem in systems employing phase modulation) 
will suffer some degradation due to the presence of phase noise''.
This is because the filtering converts phase fluctuations to amplitude variations.
It is worth mentioning that filtering is necessary before sampling to limit the variance of the noise samples.

The model (\ref{eq:dt-pn-ch}) thus does not fit\footnote{
The model would fit if a very narrow pulse was used as the pulse-shaping filter. 
Such pulses may be interesting for very large bandwidth applications.
} the channel (\ref{eq:waveform_ch_intro})
and it is not obvious whether a pre-log 1/2 is achievable.
The model that takes the effect of filtering into account is
\begin{align}
 Y_k = \Xsymb{k} H_k + N_k
 \label{eq:dt-noncoherent-fade}
\end{align}
where $\{H_k\}$ is a fading process.
The model (\ref{eq:dt-noncoherent-fade}) falls in the class of non-coherent fading channels,
i.e., the transmitter and receiver have knowledge of the distribution of the fading process $\{H_k\}$, 
but have no knowledge of its realization.
For such channels, Lapidoth and Moser showed in \cite{Lapidoth2003} that, 
at high SNR, the capacity grows \emph{double-logarithmically} with the SNR, 
when the process $\{H_k\}$ is stationary, ergodic and has finite differential entropy rate.

Rather than using a filter and sampling its output at the symbol rate,
we use a multi-sample receiver, i.e., a filter whose output is sampled many times per symbol.
We show that this receiver achieves a rate that grows \emph{logarithmically} with the SNR
if the number of samples per symbol grows with the cubic root of the SNR.
Furthermore, we show that a pre-log of 1/2 is achievable through amplitude modulation.
We also show for an approximate multi-sample model, where the approximation is to ignore the filtering effect over a sample, that phase modulation contributes 1/4 to the pre-log factor at high SNR
if the oversampling factor grows with the square root of SNR.
We corroborate the results of the high-SNR analysis with numerical simulations at finite SNR.
This paper collects and extends results presented in \cite{GhozlanISIT2013,GhozlanGLOBECOM2013,GhozlanISIT2014}.


%

Benefits of oversampling were reported for other problems.
For example, in some digital storage systems, it was shown by numerical simulations 
that oversampling can increase the information rate
\cite{PighiOversampling2007}.
In a low-pass-filter-and-limiter (nonlinear) channel, it was found that oversampling 
(with a factor of 2) offers a higher information rate
\cite{GilbertOversampling1993}.
It was demonstrated in \cite{KochOversamp2010} that doubling the sampling rate recovers some of the loss in capacity 
incurred on the bandlimited Gaussian channel with a one-bit output quantizer.
For non-coherent block fading channels, it was shown in \cite{DorpinghausOversamp2014} that
by doubling the sampling rate the information rate grows logarithmically at high SNR with a pre-log $1-1/N$ rather than
$1-Q/N$ which is the pre-log achieved by sampling at the symbol rate  
($N$ is the length of the fading block and $Q$ is the rank of the covariance matrix of the discrete-time channel gains within each block).


This paper is organized as follows.
The continuous-time model is described in Sec. \ref{sec:ct-model}
and the discrete-time models of the matched filter receiver and the multi-sample receiver are described in Sec. \ref{sec:dt-models}.
We give a lower bound on the capacity in Sec. \ref{sec:high-snr} based on amplitude modulation 
and show that the pre-log factor is at least 1/2 if the oversampling factor grows with the cubic root of the SNR.
We also give a lower bound on the rate achieved by phase modulation for the approximate multi-sample discrete-time model.
For such a model, we show that phase modulation contributes at least 1/4 to the pre-log factor, 
resulting in an overall pre-log factor of 3/4 when both amplitude and phase modulation are used.
This is achieved by using an oversampling factor that grows with the square root of the SNR.
We develop algorithms to compute tight lower bounds on the information rates
of a multi-sample receiver in Sec. \ref{sec:lower-bound}.
In Sec. \ref{sec:num-sim}, we report the results of numerical simulations 
which demonstrate the importance of increasing the oversampling factor with the SNR to achieve higher rates.
We discuss in Sec. \ref{sec:discuss} our results, related work, open problems and intuitions for our proofs.
Finally, we conclude with Sec. \ref{sec:conc}.

\section{Waveform Phase Noise Channel}
\label{sec:ct-model}
We use the following notation: 
$j=\sqrt{-1}$ , 
$^*$ denotes the complex conjugate, 
$\delta_D$ is the Dirac delta function, 
$\lceil \cdot \rceil$ is the ceiling operator.
We use $X^k$ as a shorthand for $(X_1,X_2,\ldots,X_k)$.
Suppose the transmit-waveform is $x(t)$ and 
the receiver observes 
\begin{align}
  r(t) = x(t) \ e^{j \theta(t)} + n(t)
  \label{eq:rx_waveform}
\end{align}
where $n(t)$ is a realization of a white circularly-symmetric complex Gaussian process $N(t)$ with
\begin{align}
&\mathbb{E}\left[ N(t) \right] = 0 \nonumber \\
&\mathbb{E}\left[ N(t_1) N^*(t_2) \right] = \sigma^2_N ~ \delta_D(t_2-t_1).
\label{eq:Nt_statistics}
\end{align}
The phase $\theta(t)$ is a realization of a Wiener process 
\begin{align}
  \Theta(t) = \Theta(0) + \int_0^t W(\tau) d\tau
\label{eq:Thetat}
\end{align}
where $\Theta(0)$ is uniform  on $[-\pi,\pi)$ and 
$W(t)$ is a real Gaussian process with
\begin{align}
&\mathbb{E}\left[ W(t) \right] = 0 \nonumber \\
&\mathbb{E}\left[ W(t_1) W(t_2)\right] = 2\pi \beta ~ \delta_D(t_2-t_1) 
\label{eq:Wt_statistics}
\end{align}
where $0 < \beta < \infty$.
The processes $N(t)$ and $\Theta(t)$ are independent of each other and independent of the input.
$N_0 = 2 \sigma^2_N$ is the single-sided power spectral density of the additive noise.
We define
$U(t) \equiv \exp(j \Theta(t))$.
The auto-correlation function of $U(t)$ is
\begin{align}
R_U(t_1,t_2) 
= \mathbb{E}\left[ U(t_1) U^*(t_2)\right] 
= \exp\left(- \pi \beta |t_2-t_1| \right)
\end{align}
and the power spectral density of $U(t)$ is
\begin{align}
S_U(f)
= \int_{-\infty }^{\infty} R_U(t,t+\tau) \ e^{-j 2\pi f \tau} d\tau
= \frac{1}{\pi} \frac{\beta/2}{(\beta/2)^2+ f^2}.
\end{align}
The spectrum is said to have a Lorentzian shape.
We have
$\beta = f_{\text{FWHM}} = 2 f_{\text{HWHM}}$
where 
$f_{\text{FWHM}}$ is the full-width at half-maximum
and
$f_{\text{HWHM}}$ is the half-width at half-maximum.
One also refers to $\beta$ as the \emph{linewidth} of the phase noise.

Let $T$ be the transmission interval, then
the transmitted waveforms must satisfy the power constraint
\begin{align}
	\mathbb{E}\left[\frac{1}{T} \int_0^{T} |X(t)|^2 dt \right] \leq \mathcal{P}
	\label{eq:finitesupport_waveform_power_constraint}
\end{align}
where $X(t)$ is a random process whose realization is $x(t)$.

\section{Discrete-time Models}
\label{sec:dt-models}
\subsection{Transmitter}
Let $(\xsymb{1},\xsymb{2},\ldots,\xsymb{\nsymb})$
be the codeword sent by the transmitter.
Suppose the transmitter uses a unit-energy pulse $g(t)$ whose time support is $[0,\Ts]$
where $\Ts$ is the symbol interval, which is finite and fixed.
The waveform sent by the transmitter is
\begin{align}
 x(t) = \sum_{m=1}^{\nsymb}  \xsymb{m} \ g(t-(m-1) \Ts).
\label{eq:xt_modulated_rect}
\end{align}
We remark that using pulses that are limited in time to one symbol interval simplifies the analysis
because it eliminates intersymbol interference at the receiver. 
We focus on rectangular pulses for the high SNR analysis in Sec. \ref{sec:high-snr}.
A disadvantage of a rectangular pulses is the slow decay of the sidelobes in frequency domain,
which motivates us to also consider cosine-squared pulses for numerical rate computation in Sec. \ref{sec:num-sim}.
The sidelobes of a cosine-squared pulses, 
which is equivalent to a (time-domain) raised-cosine pulse with roll-off factor 1,
decay much faster than those of a rectangular pulse.
We also remark that Martal\`{o} et al\cite{MartaloBW2015} extended some of our numerical results for bandlimited pulses.
However, their model does not include fully the filtering effects of a continuous-time model.

\subsection{Matched Filter Receiver}
\label{sec:dt-model-mf}
The received signal $r(t)$ is fed to a matched filter and 
the output $y(t)$ of the filter is sampled at the symbol rate as illustrated in Fig. \ref{fig:rx_matched_filter}.
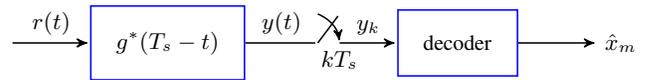
\begin{figure}
   \centering
	 \resizebox{\ifdefined\twocolumnmode 1\columnwidth \else 0.6\textwidth \fi}{!}{{\scalefont{0.85}
\begin{tikzpicture}[->,>=stealth', shorten >=1pt, auto, node distance=2.25cm, semithick]
  \tikzstyle{circ}=[circle,thick,draw=blue!75,minimum size=3pt]
  \tikzstyle{rect} = [draw=blue, rectangle, inner sep=10pt, inner ysep=10pt]
  \tikzset{text/.default=}
  
	\node (DUMMY) [] {};
  \node[rect] (RX) [right of=DUMMY] {$g^*(\Ts-t)$};
  \node[rect] (DEC) [right  of=RX,xshift=1.75cm] {decoder};
  \draw [-] (RX.east) -- +(1cm,0) node [anchor=south,xshift=-0.5cm] {$y(t)$};
  \draw [-] (RX.east) +(1cm,0) -- +(1.3cm,0.4cm) node [anchor=north,yshift=-0.4cm] {\small $k \Ts$};
  \draw [->] (RX.east) +(1cm,0.4) to [bend left=30] +(1.3cm,0cm);
  \draw [->] (RX.east) +(1cm,0) +(1.3cm,0cm) -- node [anchor=south] {$y_k$} (DEC);
  
	\node[text] (MSG_ESTIMATE) [right of=DEC] {$\hat{x}_m$};
  \path (DEC) edge node {} (MSG_ESTIMATE);
	
  \path (DUMMY) edge node {$r(t)$} (RX);

\draw [help lines];
\end{tikzpicture}
}}
	\caption{Matched filter receiver with symbol rate sampling.}
	\label{fig:rx_matched_filter}
\end{figure}
The $k$-th output sample is
\begin{align}
	y_k 
	&= r(t) \star g^*(\Ts-t) |_{t=k \Ts} \nonumber \\
	&= \xsymb{k} \ \int_{- \infty}^{\infty} e^{j \theta(\tau)} ~ |g(\tau-(k-1) \Ts)|^2 d\tau + n_k
\end{align}
where $k=1,\ldots,\nsymb$, $\star$ denotes convolution
and $n_k$ is the part of the output due to the additive noise.
Therefore, we have the discrete-time model
\begin{align}
 Y_k = \Xsymb{k} H_k + N_k
 \label{eq:dt-mfrx-full}
\end{align}
where 
\begin{align}
H_k \equiv \int_{- \infty}^{\infty} e^{j \Theta(\tau)} ~ |g(\tau-(k-1) \Ts)|^2 d\tau.
\label{eq:Hk}
\end{align}
The transmitter and receiver know the distribution of the fading process $\{H_k\}$
but have no knowledge of its realization.
We strongly suspect that $\{H_k\}$ in (\ref{eq:Hk}) is stationary, ergodic and has a finite differential entropy rate, i.e., $h(\{H_k\}) > -\infty$.
Lapidoth and Moser showed in \cite{Lapidoth2003} that, at high SNR, 
the capacity grows \emph{double-logarithmically} with the SNR under the aforementioned conditions on $\{H_k\}$.

A commonly-used approximation for $H_k$ is
\ifdefined\twocolumnmode{
\begin{align}
	&\int_{- \infty}^{\infty} e^{j \theta(\tau)} ~ |g(\tau-(k-1) \Ts)|^2 d\tau \nonumber\\
	&\approx 
	e^{j \theta(k\Ts)} \int_{- \infty}^{\infty} |g(\tau-(k-1) \Ts)|^2 d\tau
\end{align}
}\else{
\begin{align}
	\int_{- \infty}^{\infty} e^{j \theta(\tau)} ~ |g(\tau-(k-1) \Ts)|^2 d\tau
	\approx 
	e^{j \theta(k\Ts)} \int_{- \infty}^{\infty} |g(\tau-(k-1) \Ts)|^2 d\tau
\end{align}
}\fi
which yields the approximate discrete-time model
\begin{align}
	Y_k = \Xsymb{k} ~ e^{j \Theta(k \Ts)} + N_k.
	\label{eq:dt-mfrx-approx}
\end{align}
We refer to (\ref{eq:dt-mfrx-full}) and (\ref{eq:dt-mfrx-approx}) as the symbol-rate model with filtering and without filtering, respectively.
We also refer to (\ref{eq:dt-mfrx-approx}) as the approximate symbol-rate model.
We do not claim that this approximation is always accurate.
We introduce the approximate model 
to examine its validity by comparing it with 
the discrete-time model based on the multi-sample receiver.
We remark that the differential entropy of $\{\Theta(k \Ts)\}$ is finite 
because the linewidth $\beta$ of $\Theta(t)$ and the symbol interval $\Ts$ are finite
and therefore, the pre-log of the model (\ref{eq:dt-mfrx-approx}) is 1/2.

\subsection{Multi-sample Receiver}
\label{sec:dt-model-multisample}
Consider next a multi-sample receiver (see Fig. \ref{fig:rx_oversamp}).
Let $L$ be the number of samples per symbol ($L \geq 1$)
and define the sample interval as
\begin{align}
\Delta = \frac{\Ts}{L}.
\label{eq:Delta_def}
\end{align}
The received waveform $r(t)$ is filtered using an integrator over a sample interval to give the output signal
\begin{align}
y(t) 
&= \int_{t - \Delta}^{t} r(\tau) \ d\tau.
\label{eq:yt}
\end{align}
The signal $y(t)$ is a realization of a random process $Y(t)$ that
is sampled at $t = k \Delta$,
$k = 1, \ldots, \nsamp = \nsymb L$.
\begin{figure}
   \centering
	 \resizebox{\ifdefined\twocolumnmode 1\columnwidth \else 0.6\textwidth \fi}{!}{{\scalefont{0.85}
\begin{tikzpicture}[->,>=stealth', shorten >=1pt, auto, node distance=2.25cm, semithick]
  \tikzstyle{circ}=[circle,thick,draw=blue!75,minimum size=3pt]
  \tikzstyle{rect} = [draw=blue, rectangle, inner sep=10pt, inner ysep=10pt]
  \tikzset{text/.default=}
  
	\node (DUMMY) [] {};
  \node[rect] (RX) [right of=DUMMY] {$\int$};	
  \node[rect] (DEC) [right  of=RX,xshift=1.5cm] {decoder};
  \draw [-] (RX.east) -- +(1cm,0) node [anchor=south,xshift=-0.5cm] {$y(t)$};
  \draw [-] (RX.east) +(1cm,0) -- +(1.3cm,0.4cm) node [anchor=north,yshift=-0.4cm] {\small $k \Ts/L$};
  \draw [->] (RX.east) +(1cm,0.4) to [bend left=30] +(1.3cm,0cm);
  \draw [->] (RX.east) +(1cm,0) +(1.3cm,0cm) -- node [anchor=south] {$y_k$} (DEC);
  
	\node[text] (MSG_ESTIMATE) [right of=DEC] {$\hat{x}_m$};
  \path (DEC) edge node {} (MSG_ESTIMATE);
	
  \path 
	(DUMMY) edge node {$r(t)$} (RX)  
;

\draw [help lines];
\end{tikzpicture}
}}
	 \caption{Multi-sample receiver.}
	 \label{fig:rx_oversamp}
\end{figure}
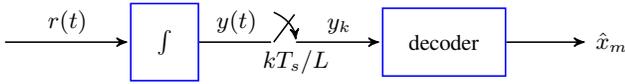
Hence, we have the discrete-time model in which the $k$-th output sample is
\begin{align}
Y_k = \Xsymb{\lceil k/L \rceil} \Delta \ e^{j \Theta_k} \ F_k + N_k
\label{eq:Yk}
\end{align}
where
$Y_k \equiv Y(k \Delta)$, 
$\Theta_k \equiv \Theta( (k-1) \Delta )$,
\begin{align}
F_k \equiv \frac{1}{\Delta} \int_{(k-1) \Delta}^{k \Delta} 
g\left(\tau-\left(\left\lceil\frac{k}{L}\right\rceil -1\right) \Ts\right)
e^{j(\Theta(\tau)-\Theta_k)} \ d\tau
\label{eq:Fk_def_tlim}
\end{align}
and
\begin{align}
N_k 
&\equiv \int_{(k-1) \Delta}^{k \Delta} N(\tau) \ d\tau.
\label{eq:Nk_def}
\end{align}
The process $\{N_k\}$ is an i.i.d. circularly-symmetric complex Gaussian process with 
\begin{align}
\mathbb{E}[ N_k] &= 0
\label{eq:N_mean}
\\
\mathbb{E}[ N_k N_\ell^* ] &= \sigma^2_N \Delta \ \delta_K[\ell-k]
\label{eq:N_covar}
\\
\mathbb{E}[ N_k N_\ell ] &= 0
\label{eq:N_pseudocovar}
\end{align}
where $\delta_K[\cdot]$ is the Kronecker delta.
The process $\{\Theta_k\}$ is the discrete-time Wiener process
\begin{align}
  \Theta_k = \Theta_{k-1} + W_k 
\end{align}
for $k = 2, \ldots, n$,
where
$\Theta_1$ is uniform on $[-\pi,\pi)$ and
$\{W_k\}$ is an i.i.d. real Gaussian process with mean $0$ and $\mathbb{E}[ |W_k|^2 ] = 2\pi \beta \Delta$, i.e., 
the probability density function (pdf) of $W_k$ is $p_{W_k}(w) = G(w;0,\sigma^2_W)$ where
\begin{align}
 G(w;\mu,\sigma^2) = \frac{1}{\sqrt{2\pi \sigma^2}} \exp\left( - \frac{(w-\mu)^2}{2 \sigma^2} \right)
\end{align}
and $\sigma^2_W = 2 \pi \beta \Delta$.
The random variable $(W_k \mod 2\pi)$  is a \emph{wrapped Gaussian} and its pdf is
\begin{align}
 p_{W}(w;\sigma^2) = \sum_{i=-\infty}^{\infty} G(w-2 i \pi;0,\sigma^2).
\end{align}
Moreover, $\{F_k\}$ and $\{W_k\}$ are independent of $\{N_k\}$ but not independent of each other.
Equations (\ref{eq:finitesupport_waveform_power_constraint}) and (\ref{eq:xt_modulated_rect})
imply the power constraint 
\begin{align}
	\frac{1}{\nsymb} \sum_{m=1}^{\nsymb} \mathbb{E}[|\Xsymb{m}|^2] \leq P \equiv \mathcal{P} \Ts.
	\label{eq:dt_power_constraint}
\end{align}
The signal-to-noise ratio $\SNR$ is defined as
\begin{align}
	\SNR \equiv \frac{P}{\sigma^2_N \Ts} = \frac{\mathcal{P}}{\sigma^2_N}.
\end{align}


\subsection{Information Rates}
\label{sec:info-rates-def}
We use $\Y^k$ as a shorthand for $(\Y_1,\Y_2,\ldots,\Y_k)$ where
\begin{align} \Y_k \equiv (Y_{(k-1) L + 1},Y_{(k-1) L + 2},\ldots,Y_{(k-1) L + L}). \label{eq:Yveck_def} \end{align}
The capacity of a point-to-point channel is given by
\begin{align}
C(\SNR) = \lim_{\nsymb \rightarrow \infty} \frac{1}{\nsymb} \sup I(X^\nsymb;\Y^\nsymb)
\label{eq:capacity-def}
\end{align}
where the supremum is over all joint distributions 
of the input symbols satisfying the power constraint.
For a given input distribution, the achievable rate $R$ is
\begin{align}
R(\SNR) = I(X;Y) \equiv \lim_{\nsymb \rightarrow \infty} \frac{1}{\nsymb} I(X^\nsymb;\Y^\nsymb).
\label{eq:I_X_Y-def}
\end{align}

\section{High-SNR Information Rates}
\label{sec:high-snr}

In this section, we study rectangular pulses, i.e., we consider
\begin{align}
 g(t) \equiv \left\{ 
  \begin{array}{ll}
  \sqrt{1/\Ts},	& 0 \leq t <\Ts, \\
  0,			& \text{otherwise}.
  \end{array}
 \right.
\label{eq:rect_pulse_def}
\end{align}
It follows that $F_k$ in (\ref{eq:Fk_def_tlim}) becomes
\begin{align}
F_k \equiv \frac{1}{\Delta} \int_{(k-1) \Delta}^{k \Delta} e^{j(\Theta(\tau)-\Theta_k)} \ d\tau.
\label{eq:Fk_def}
\end{align}

We define $X_A \equiv |X|$ and $\Phi_X \equiv \arg(X)$ 
where $\arg(\cdot)$ denotes the phase angle (also called the argument) of a complex number.
We decompose the mutual information using the chain rule into two parts:
\begin{align}
I(X^\nsymb;\Y^\nsymb) 
&= I(X_{A}^\nsymb;\Y^\nsymb) + I(\Phi_X^\nsymb;\Y^\nsymb|X_{A}^\nsymb)
\label{eq:I_X1n_Y1n}
\end{align}
where $X_A^\nsymb = (|X_1|,\ldots,|X_\nsymb|)$ and $\Phi_X^\nsymb = (\arg{X_1},\ldots,\arg{X_\nsymb})$.
Define
\begin{align}
I(X_{A};Y) 
&\equiv \lim_{\nsymb \rightarrow \infty} \frac{1}{\nsymb} I(X_{A}^\nsymb;\Y^\nsymb)
\label{eq:I_XA_Y-def}
\end{align}
and
\begin{align}
I(\Phi_{X};Y|X_{A})
&\equiv \lim_{\nsymb \rightarrow \infty} \frac{1}{\nsymb} I(\Phi_{X}^\nsymb;\Y^\nsymb|X_{A}^\nsymb) .
\label{eq:I_PhiX_Y|XA-def}
\end{align}
It follows from (\ref{eq:I_X_Y-def})--(\ref{eq:I_PhiX_Y|XA-def}) that
\begin{align}
I(X;Y) = I(X_A;Y) + I(\Phi_{X};Y|X_{A}) .
\label{eq:I_X_Y}
\end{align}
The first term represents the contribution of the amplitude modulation while 
the second term represents the contribution of the phase modulation.
The results of this section are summarized in 
Theorem \ref{theorem:oversamp-fullmodel} and Theorem \ref{theorem:oversamp-approxmodel}
which address the contribution of amplitude modulation and phase modulation to the information rate, respectively.

Before we proceed to the theorems, we introduce an important result that we use multiple times in the paper.
Let $X$ and $Y$ be the input and output of a channel, respectively and 
let $q(\cdot|\cdot)$ be a valid conditional pdf, 
which represents an ``auxiliary channel''.
The following auxiliary-channel lower bound holds:
\begin{align}
I(X;Y)
&\geq \underline{I}(X;Y) 
= \mathbb{E}[\log{q_{Y|X}(Y|X)-\log{q_Y(Y)}}]
\label{eq:I_AuxCh}
\end{align}
where $q_Y(\cdot)$ is the output pdf obtained by connecting the true input source to the auxiliary channel.
Note that $\mathbb{E}[\cdot]$ is the expectation according to the \emph{true} distribution of $X$ and $Y$.
The lower bound is the rate achieved by a mismatched decoder.
See \cite[pp. 6-7]{ScarlettPHD2015} for a recent comprehensive review on mismatched decoding.

\begin{theorem}[Amplitude Modulation]
\label{theorem:oversamp-fullmodel}
Consider the discrete-time model in (\ref{eq:Yk}) with $F_k$ as defined in (\ref{eq:Fk_def}).
If $L = \lceil \sqrt[3]{(\beta\Ts)^2 \SNR} \rceil$ and
the process $\{|X_k|\}$ is i.i.d. such that 
$(|X_k|^2-P/2)$ is exponentially distributed with mean $P/2$,
then
\begin{align}
\lim_{\SNR \rightarrow \infty} I(X_A;Y) - \frac{1}{2} \log{\SNR}
&\geq - 2  - \frac{1}{2} \log(8 \pi)
- \frac{\pi^2}{45}.
\label{eq:I_XA_Y_fullmodel_theorem}
\end{align}
\end{theorem}
\begin{proof} See Appendix \ref{sec:amplitude-lower-bound}. 
The lower bound in (\ref{eq:I_XA_Y_fullmodel_theorem}) is achieved by 
the double-filtering receiver structure shown in Fig. \ref{eq:rx_double_filtering}.
We borrow the term ``double-filtering'' from \cite{Foschini1988IT}. 
\end{proof}

\begin{figure*}
   \centering
	 \resizebox{\ifdefined\twocolumnmode 0.8\textwidth \else \textwidth \fi}{!}{{\scalefont{0.85}
\begin{tikzpicture}[->,>=stealth', shorten >=1pt, auto, node distance=2.25cm, semithick]
  \tikzstyle{circ}=[circle,thick,draw=blue!75,minimum size=3pt]
  \tikzstyle{rect} = [draw=blue, rectangle, inner sep=10pt, inner ysep=10pt]
  \tikzset{text/.default=}
  
  \node (DUMMY) [] {};
  \node[rect] (RX) [right of=DUMMY] {$\int$};
	\node[rect] (SQR) [right  of=RX,xshift=1.5cm] {$(\cdot)^2$};
  \draw [-] (RX.east) -- +(1cm,0) node [anchor=south,xshift=-0.5cm] {$y(t)$};
  \draw [-] (RX.east) +(1cm,0) -- +(1.3cm,0.4cm) node [anchor=north,yshift=-0.4cm] {\small $k \Ts/L$};
  \draw [->] (RX.east) +(1cm,0.4) to [bend left=30] +(1.3cm,0cm);
  \draw [->] (RX.east) +(1cm,0) +(1.3cm,0cm) -- node [anchor=south] {$y_k$} (SQR);
	\node[rect] (ACC) [right  of=SQR,xshift=-0.5cm] {$\Sigma$};
  \node (RXTXT) [below of=RX,yshift=+1cm,text width=2cm,align=center] {integrate and dump};
  \node (ACCTXT) [below of=ACC,yshift=+1cm,text width=2cm,align=center] {digital filter};
	\node[rect] (DS) [right  of=ACC,xshift=-0.25cm] {$L \downarrow$};
  \node (DSTXT) [below of=DS,yshift=+1cm,text width=2cm,align=center] {downsample};
	\node[rect] (DEC) [right  of=DS,xshift=0.5cm,text width=1.5cm,align=center] {Amplitude decoder};
  
  \path 
  (DUMMY) edge node {$r(t)$} (RX)  
  (SQR) edge node {} (ACC)
	(ACC) edge node {} (DS)
  (DS) edge node {$v_m$} (DEC)
;

	\node[text] (MSG_ESTIMATE) [right of=DEC] {$\hat{x}_{A,m}$};
  \path (DEC) edge node {} (MSG_ESTIMATE);
	
\draw [help lines];
\end{tikzpicture}
}}
	 \caption{Double filtering receiver.}
	 \label{eq:rx_double_filtering}
\end{figure*}
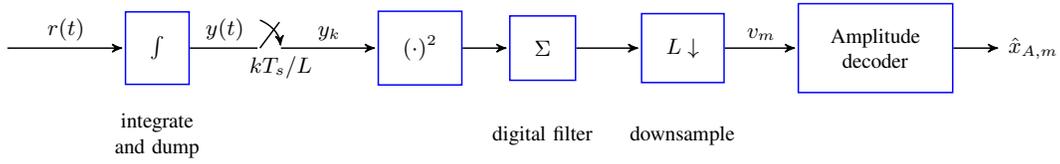

Theorem \ref{theorem:oversamp-fullmodel} implies that the capacity pre-log satisfies
\begin{align}
\lim_{\SNR \rightarrow \infty} \frac{C(\SNR)}{\log{\SNR}} \geq \frac{1}{2}
\end{align}
which follows from (\ref{eq:I_XA_Y_fullmodel_theorem}) and $C(\SNR) \geq I(X;Y) \geq I(X_A;Y)$.
The capacity thus grows logarithmically at high SNR with a pre-log factor of at least 1/2
if the oversampling factor $L$ grows with the cubic root of SNR.

Consider next a simplified model 
that does not include the effect of filtering modeled by $\{F_k\}$.
More specifically, the $k$th output of the simplified approximate model is
\begin{align}
\Psi_k = X_{\lceil k/L \rceil} \Delta \ e^{j \Theta_k} + N_k
\label{eq:Yk_approx}
\end{align}
where $X_k$ is $k$th input \emph{symbol} and $\{\Theta_k\}$ and $\{N_k\}$ are the same processes defined in Sec. \ref{sec:dt-model-multisample}.
\begin{theorem}[Phase Modulation]
\label{theorem:oversamp-approxmodel}
Consider the discrete-time model in (\ref{eq:Yk_approx}).
If $L = \lceil \sqrt{\pi \beta\Ts \SNR}/2 \rceil$ and
the input $X^\nsymb$ is i.i.d. with $\arg(X_k)$ independent of $|X_k|$ for $k=1,\ldots,\nsymb$ such that 
$\arg(X_k)$ is uniformly distributed over $[-\pi,\pi)$ and
$(|X_k|^2-P/2)$ is exponentially distributed with mean $P/2$,
then
\begin{align}
\lim_{\SNR \rightarrow \infty} I(X_{A};\Psi) - \frac{1}{2} \log{\SNR}
\geq - 2  - \frac{1}{2} \log(8 \pi)
\label{eq:I_XA_Y_approxmodel_theorem}
\end{align}
and
\begin{align}
\lim_{\SNR \rightarrow \infty} I(\Phi_{X};\Psi|X_{A}) - \frac{1}{4} \log{\SNR}
\geq - \frac{1}{4} \log(64 \pi \beta\Ts) - \frac{1}{2}.
\label{eq:I_PhiX_Y|XA_approxmodel_theorem}
\end{align}

\end{theorem}
\begin{proof} See Appendix \ref{sec:approxmodel-phase}. 
The lower bounds (\ref{eq:I_XA_Y_approxmodel_theorem}) and (\ref{eq:I_PhiX_Y|XA_approxmodel_theorem})
are achieved by the receiver structure shown in Fig. \ref{fig:two-stage}.
\end{proof}
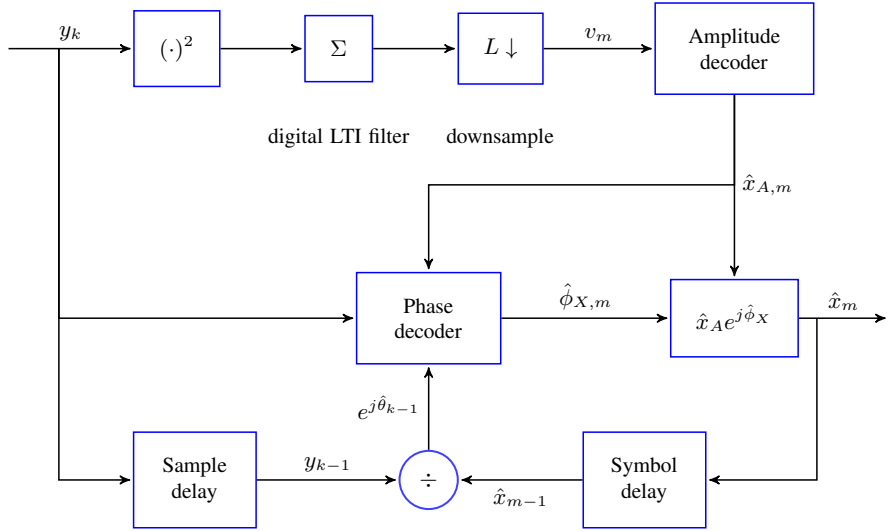
\begin{figure*}
   \centering
	 \resizebox{\ifdefined\twocolumnmode 0.75\textwidth \else 0.75\textwidth \fi}{!}{{\scalefont{0.85}
\begin{tikzpicture}[->,>=stealth', shorten >=1pt, auto, node distance=2.25cm, semithick]
  \tikzstyle{circ}=[circle,thick,draw=blue!75,minimum size=3pt]
  \tikzstyle{rect} = [draw=blue, rectangle, inner sep=10pt, inner ysep=10pt]
  
  \node (DUMMY) [] {};
	\node[rect] (SQR) [right  of=DUMMY,xshift=1.5cm] {$(\cdot)^2$};
  \draw [->] (DUMMY.east) +(1cm,0) +(1.3cm,0cm) -- node [anchor=south] {$y_k$} (SQR);
  \node[rect] (ACC) [right  of=SQR,xshift=-0cm] {$\Sigma$};
  \node (ACCTXT) [below of=ACC,yshift=+1cm,text width=2cm,align=center] {digital LTI filter};
	\node[rect] (DS) [right  of=ACC,xshift=-0cm] {$L \downarrow$};
  \node (DSTXT) [below of=DS,yshift=+1cm,text width=2cm,align=center] {downsample};
  \node[rect] (DEC) [right  of=DS,xshift=1cm,text width=1.5cm,align=center] {Amplitude decoder};
  \node[rect] (DEC_PHASE) [below  of=ACC,yshift=-1.5cm,xshift=1.25cm,text width=1.3cm,align=center] {Phase decoder};	
  \node[circ] (THETA_HAT) [below of=DEC_PHASE,text width=0.5cm,align=center] {$\div$};	
  \node[rect] (DELAY_Y) [left of=THETA_HAT,xshift=-1.00cm,text width=1cm,align=center] {Sample delay};
  
	\draw [->] (DUMMY.east) +(1cm,0) +(2cm,0cm) |- (DELAY_Y);
  \node[rect] (DELAY_XHAT) [right of=THETA_HAT,xshift=0.75cm,text width=1cm,align=center] {Symbol delay};
  \draw [->] (DUMMY.east) +(1cm,0) +(2cm,0cm) |- (DEC_PHASE);
  \node[rect] (X_HAT) [below of=DEC,yshift=-1.5cm] {$\hat{x}_{A} e^{j \hat{\phi}_X}$};
  \node (DUMMYSINK) [right of=X_HAT] {};
  \draw [->] (X_HAT.east)+(0.25cm,0) |- (DELAY_XHAT.east);
  \draw [->] (DEC.south) |- +(0cm,-1.25cm) node [anchor=west] {$\hat{x}_{A,m}$} -| (DEC_PHASE.north);

  \path 
  (SQR) edge node {} (ACC)
	(ACC) edge node {} (DS)
  (DS) edge node {$v_m$} (DEC)

	(DELAY_Y) edge node {$y_{k-1}$} (THETA_HAT)
 	(DELAY_XHAT) edge node {$\hat{x}_{m-1}$} (THETA_HAT)
	(THETA_HAT) edge node {$e^{j \hat{\theta}_{k-1}}$} (DEC_PHASE)
	(DEC_PHASE) edge node {$\hat{\phi}_{X,m}$} (X_HAT)
	(DEC) edge node {} (X_HAT)
	(X_HAT) edge node {$\hat{x}_m$} (DUMMYSINK)
	
;

\draw [help lines];
\end{tikzpicture}
}}
	 \caption{Two-stage decoding in a multi-sample receiver. 
	 The upper branch is part of  the double filtering receiver for amplitude decoding.
	 The lower branch is a phase-noise tracker which outputs an estimate 
	 $\hat{\theta}_{k-1}$ of the phase noise where $\hat{\theta}_{k-1} = \arg(y_{k-1}/\hat{x}_{m-1})$ where $m = \lceil k/L \rceil$ (cf. (\ref{eq:Theta_hat})).
	 The phase decoder uses the estimate of the phase noise $\hat{\theta}_{k-1}$ and the estimated amplitude $\hat{x}_{A,m}$
	 to produce an estimate $\hat{\phi}_{X,m}$ of the phase of $x_m$.
	 }
	 \label{fig:two-stage}
\end{figure*}

Theorem \ref{theorem:oversamp-approxmodel} implies that
the capacity pre-log for the approximate model (\ref{eq:Yk_approx}) satisfies
\begin{align}
\lim_{\SNR \rightarrow \infty} \frac{C(\SNR)}{\log{\SNR}} \geq \frac{3}{4}
\label{eq:prelog_approxmodel}
\end{align}
which follows from (\ref{eq:I_XA_Y_approxmodel_theorem})-(\ref{eq:I_PhiX_Y|XA_approxmodel_theorem}) by using $C(\SNR) \geq I(X;\Psi) = I(X_{A};\Psi) + I(\Phi_{X};\Psi|X_{A})$.
This shows that the capacity grows logarithmically at high SNR with a pre-log factor of at least 3/4.
We remark that Barletta and Kramer \cite{BarlettaISIT2015} extended the result of Theorem \ref{theorem:oversamp-approxmodel}
to the full model (\ref{eq:Yk}) that includes the filtering effect, 
i.e., they showed that $C(\SNR) \geq 3/4$ for (\ref{eq:Yk}) if $L = \lceil \sqrt{\SNR} \rceil$.
More generally, for $L=\lceil \SNR^\alpha \rceil$, they showed for the full model that
\begin{align}\label{eq:prelog_vs_L}
 \lim_{\SNR\rightarrow\infty} \frac{C(\SNR)}{\log(\SNR)} \geq 
 \left\{
 \begin{array}{ll}
  2\alpha, 			 & 0<\alpha\le 1/3 \\
  (1+\alpha)/2,	 & 1/3\le\alpha\le 1/2 \\
  3/4,					 & 1/2 \le \alpha < 1 .
 \end{array}\right.
\end{align}



\section{Rate Computation at Finite SNR}
\label{sec:lower-bound}
We develop algorithms to compute the information rate numerically.
We use $\Xsamp_k$ to denote the $k$th input \emph{sample} which is defined by
\begin{align}
\Xsamp_k 
= X_{\lceil k/L \rceil} ~ g\left( (k \text{ mod } L) \Delta\right).
\label{eq:Xk_def}
\end{align}
The information rate $I(X;Y)$ in (\ref{eq:I_X_Y-def}) can be expressed as
\begin{align}
  I(X;Y) 
  = \lim_{\nsymb \rightarrow \infty} \frac{1}{\nsymb} I(\Xsamp^n;Y^n)
\label{eq:IXY}
\end{align}
which follows from (\ref{eq:Xk_def}) and (\ref{eq:Yveck_def}).
One difficulty in evaluating (\ref{eq:IXY}) is that the joint distribution of $\{F_k\}$ and $\{W_k\}$ is not available in closed form.
Even the distribution of $F_k$ is not available in closed form
(there is an approximation for small linewidth, see (16) in \cite{Foschini1988Comm}).
However, we can numerically compute tight lower bounds on $I(X;Y)$ by using the auxiliary-channel technique (cf. (\ref{eq:I_AuxCh})).
We use the following algorithm \cite{Arnold2006}.
\begin{enumerate}
\item Sample a long sequence $(x^n, y^n)$ according to the \emph{true} joint distribution of $\Xsamp^n$ and $Y^n$; 
\item Compute $q_{Y^n|\Xsamp^n}(y^n|x^n)$
and
\begin{align}
 q_{Y^n}(y^n) = \sum_{\tilde{x}^n} p_{\Xsamp^n}(\tilde{x}^n) q_{Y^n|\Xsamp^n}(y^n|\tilde{x}^n)
\end{align}
where $p_{\Xsamp^n}$ is the true distribution of $\Xsamp^n$;
\item Estimate $\underline{I}(X;Y)$ using
\begin{align}
  \underline{I}(X;Y) \approx 
  \frac{1}{\nsymb}\log\left( \frac{q_{Y^n|\Xsamp^n}(y^n|x^n)}{q_{Y^n}(y^n)} \right).
\end{align}
\end{enumerate}

\paragraph*{Auxiliary Channel I}
Consider the auxiliary channel
\begin{align}
\Psi_k 
= \Xsamp_k \Delta \ e^{j \Theta_k} + N_k
\label{eq:Psik_def}
\end{align}
where $\{\Theta_k\}$ and $\{N_k\}$ are defined in Sec. \ref{sec:dt-model-multisample}
and
$\Xsamp_k$ is defined by (\ref{eq:Xk_def}).
The channel output $\Psi$ is the same as $Y$ in (\ref{eq:Yk}) \emph{except} that $F_k$ is replaced by 1 or more generally with $g\left( (k \mod L) \Delta\right)$.
The channel is described by
the conditional distribution
\begin{align}
  p_{\Psi^n|\Xsamp^n}(y^n|x^n) 
  &=\int_{\theta^n} p_{\Theta^n,\Psi^n|\Xsamp^n}(\theta^n,y^n|x^n) \ d\theta^n
\end{align}
where
\begin{align}
  &p_{\Theta^n,\Psi^n|\Xsamp^n}(\theta^n,y^n|x^n) \nonumber \\ \qquad
  &= \prod_{k=1}^{n} p_{\Theta_k|\Theta_{k-1}}(\theta_k|\theta_{k-1}) \ 
     p_{\Psi|\Xsamp,\Theta}(y_k|x_k,\theta_k)
  \label{eq:p_Theta_Psi|X}
\end{align}
with
\begin{align}
p_{\Theta_k|\Theta_{k-1}}(\theta|\tilde{\theta}) = 
\left\{
\begin{array}{ll}
p_{W}(\theta-\tilde{\theta}; \sigma^2_W), 	& k \geq 2 \\
1/(2\pi), 					& k = 1 
\end{array}
\right.
\end{align}
and
\begin{align}
p_{\Psi|\Xsamp,\Theta}(y|x,\theta)
= \frac{1}{\pi \sigma^2_N \Delta} \exp\left(- \frac{\left|y - x \Delta ~ e^{j \theta} \right|^2}{\sigma^2_N \Delta} \right).
\label{eq:p_Y|XTheta}
\end{align}
The channel $p_{\Psi^n|\Xsamp^n}$ has continuous states $\theta^n$,
which implies that step 2 of the algorithm can not be computed directly.
To compute the integrals, we quantize the states (the phase noise).
Finer quantization results in more accurate results 
but is also more computationally intensive.
We approach this issue by first computing an estimate of the rate at a small number of quantization levels,
then we increase the number of quantization levels as long as there is a significant difference between the estimated rates. 
When there is no signification change in the rate estimates, 
we take that to be a (lower bound on the) rate with the auxiliary channel $\Psi_k$.
This approach is demonstrated in Fig. \ref{fig:16qam_fhwhm125_sqrtx_optrx_osr} and Fig. \ref{fig:16qam_fhwhm125_cos2tx_optrx_osr}, both in Sec. \ref{sec:num-sim}.
However, in Fig. \ref{fig:16psk_fhwhm12_sqrtx_optrx_osr} and Fig. \ref{fig:raheli_comparison}, 
we show only the final result after finding the appropriate number of quantization levels.
Next, we describe formally the quantization of the states and the rate computation in that case.

\paragraph*{Auxiliary Channel II}
We use the following auxiliary channel for the numerical simulations:
\begin{align}
\Upsilon_k 
= \Xsamp_k \Delta \ e^{j S_k} + N_k.
\label{eq:Upsilonk_def}
\end{align}
The conditional probability is
\begin{align}
  p_{\Upsilon^n|\Xsamp^n}(y^n|x^n) 
  &=\sum_{s^n \in \mathcal{S}^n} p_{S^n,\Upsilon^n|\Xsamp^n}(s^n,y^n|x^n)
\end{align}
where $\mathcal{S}$ is a \emph{finite} set and
\begin{align}
  &p_{S^n,\Upsilon^n|\Xsamp^n}(s^n,y^n|x^n) \nonumber \\ \qquad
  &= \prod_{k=1}^{n} p_{S_k|S_{k-1}}(s_k|s_{k-1}) \ 
     p_{\Psi|\Xsamp,\Theta}(y_k|x_k,s_k)
  \label{eq:p_S_Upsilon|X}
\end{align}
where
\begin{align}
p_{S_k|S_{k-1}}(s|\tilde{s}) = 
\left\{
\begin{array}{ll}
Q(s|\tilde{s}), 			& k \geq 2 \\
1/|\mathcal{S}|,			& k = 1. 
\end{array}
\right.
\end{align}

Next, we describe our choice of $\mathcal{S}$ and $Q(\cdot|\cdot)$.
We  partition $[-\pi,\pi)$ into $S$ intervals with equal lengths
and pick the mid points of these intervals to be the elements of $\mathcal{S}$,
i.e., we have
\begin{align}
\mathcal{S} = \left\{ \hat{s}_i : i = 1,\ldots,S \right\}
\text{ where }
\hat{s}_i = i \frac{2 \pi}{S} - \frac{\pi}{S} - \pi
.
\end{align}
The state transition probability $Q(\cdot|\cdot)$ is chosen similar to \cite{BarlettaLB} and \cite{BarlettaUB}:
\begin{align}
Q(s|\tilde{s}) = 
\frac{2\pi}{S}
{\int_{(\phi,\tilde{\phi}) \in \mathcal{R}(s) \times \mathcal{R}(\tilde{s})} 
p_{W}(\phi-\tilde{\phi};\sigma^2_W) \ d\phi d\tilde{\phi}}
\end{align}
where $\mathcal{R}(s) = [ s- {\pi}/{S} , s + {\pi}/{S} )$, i.e.,
$\mathcal{R}(s)$ is the interval whose midpoint is $s$.
The larger $S$ and $L$ are, the better the auxiliary channel (\ref{eq:Upsilonk_def}) approximates the actual channel (\ref{eq:Yk}).
We remark that the auxiliary channel gives a \emph{valid} lower bound on $I(X;Y)$ even for small $S$ and $L$.

\subsection{Computing The Conditional Probability}
Suppose the input $\Xsamp^n$ has the distribution $p_{\Xsamp^n}$.
A Bayesian network for $\Xsamp^n,S^n,\Upsilon^n$ is shown in Fig. \ref{fig:bayesian_net_general}.
\begin{figure}
\centering
\resizebox{\ifdefined\twocolumnmode \columnwidth \else 0.75\textwidth \fi}{!}{\begin{tikzpicture}[->,>=stealth',shorten >=1pt,auto,node distance=1.5cm,semithick]
  \tikzstyle{var}=[circle,thick,draw=blue!75,fill=blue!20,minimum size=3pt]

  \node[var]  (S1) [		label=-45:$S_1$] {};
  \node[var]  (S2) [right of=S1,label=-45:$S_2$] {};
  \node[var]  (S3) [right of=S2,label=-45:$S_3$] {};
  \node[var]  (S4) [right of=S3,label=-45:$S_4$] {};
  \node[var]  (S5) [right of=S4,label=-45:$S_5$] {};
  \node[var]  (S6) [right of=S5,label=-45:$S_6$] {};
  \node[var]  (S7) [right of=S6,label=-45:$S_7$] {};
  \node[var]  (S8) [right of=S7,label=-45:$S_8$] {};
  \node[var]  (S9) [right of=S8,label=-45:$S_9$] {};

  \node[var]  (Y1) [above of=S1,label=right:$\Upsilon_1$] {};
  \node[var]  (Y2) [right of=Y1,label=right:$\Upsilon_2$] {};
  \node[var]  (Y3) [right of=Y2,label=right:$\Upsilon_3$] {};
  \node[var]  (Y4) [right of=Y3,label=right:$\Upsilon_4$] {};
  \node[var]  (Y5) [right of=Y4,label=right:$\Upsilon_5$] {};
  \node[var]  (Y6) [right of=Y5,label=right:$\Upsilon_6$] {};
  \node[var]  (Y7) [right of=Y6,label=right:$\Upsilon_7$] {};
  \node[var]  (Y8) [right of=Y7,label=right:$\Upsilon_8$] {};
  \node[var]  (Y9) [right of=Y8,label=right:$\Upsilon_9$] {};

  \node[var]  (X1) [above of=Y1,label=-45:$\Xsamp_1$] {};
  \node[var]  (X2) [right of=X1,label=-45:$\Xsamp_2$] {};
  \node[var]  (X3) [right of=X2,label=-45:$\Xsamp_3$] {};
  \node[var]  (X4) [right of=X3,label=-45:$\Xsamp_4$] {};
  \node[var]  (X5) [right of=X4,label=-45:$\Xsamp_5$] {};
  \node[var]  (X6) [right of=X5,label=-45:$\Xsamp_6$] {};
  \node[var]  (X7) [right of=X6,label=-45:$\Xsamp_7$] {};
  \node[var]  (X8) [right of=X7,label=-45:$\Xsamp_8$] {};
  \node[var]  (X9) [right of=X8,label=-45:$\Xsamp_9$] {};

  \node[var]  (W)  [above of=X5] {};

  \path 
  (S1) edge node {} (S2)
  (S2) edge node {} (S3)
  (S3) edge node {} (S4)
  (S4) edge node {} (S5)
  (S5) edge node {} (S6)
  (S6) edge node {} (S7)
  (S7) edge node {} (S8)
  (S8) edge node {} (S9)
  (X1) edge node {} (Y1)
  (X2) edge node {} (Y2)
  (X3) edge node {} (Y3)
  (X4) edge node {} (Y4)
  (X5) edge node {} (Y5)
  (X6) edge node {} (Y6)
  (X7) edge node {} (Y7)
  (X8) edge node {} (Y8)
  (X9) edge node {} (Y9)
  (S1) edge node {} (Y1)
  (S2) edge node {} (Y2)
  (S3) edge node {} (Y3)
  (S4) edge node {} (Y4)
  (S5) edge node {} (Y5)
  (S6) edge node {} (Y6)
  (S7) edge node {} (Y7)
  (S8) edge node {} (Y8)
  (S9) edge node {} (Y9)
  (W) edge node {} (X1)
  (W) edge node {} (X2)
  (W) edge node {} (X3)
  (W) edge node {} (X4)
  (W) edge node {} (X5)
  (W) edge node {} (X6)
  (W) edge node {} (X7)
  (W) edge node {} (X8)
  (W) edge node {} (X9)
;
\end{tikzpicture}}
\caption{Bayesian network for ${\Xsamp^n,S^n,\Upsilon^n}$ for $n=9$.}
\label{fig:bayesian_net_general}
\end{figure}
The probability $p_{\Upsilon^n|\Xsamp^n}(y^n|x^n)$ can be computed using
\begin{align}
 p_{\Upsilon^n|\Xsamp^n}(y^n|x^n) 
 = \sum_{s \in \mathcal{S}} \rho_n(s)
\end{align}
where we recursively compute 
\begin{align}
  \rho_{k}(s) 
  & \equiv p_{S_k,\Upsilon^k|\Xsamp^n}(s,y^k|x^n) \label{eq:rho_def} \\
  &\stackrel{(a)}{=}  \sum_{\tilde{s} \in \mathcal{S}} p_{S_{k-1},S_k,\Upsilon^k|\Xsamp^n}(\tilde{s},s,y^k|x^n) \nonumber \\
  &\stackrel{(b)}{=}  \sum_{\tilde{s} \in \mathcal{S}} 
      \rho_{k-1}(\tilde{s}) \ 
      p_{S_k,\Upsilon_k|S_{k-1},\Upsilon^{k-1},\Xsamp^n}(s,y_k|\tilde{s},y^{k-1},x^n) \nonumber \\
  &=  \sum_{\tilde{s} \in \mathcal{S}} 
      \rho_{k-1}(\tilde{s}) \ 
      Q(s|\tilde{s}) \ 
      p_{\Psi|\Xsamp,\Theta}(y_k|x_k,s)
  \label{eq:rho_recursive_general}
\end{align}
with the initial value
$\rho_{0}(s) = 1/|\mathcal{S}|$.
Step $(a)$ is a marginalization, 
$(b)$ follows from Bayes' rule and the definition of $\rho_k$ in (\ref{eq:rho_def}), 
while (\ref{eq:rho_recursive_general}) follows from the structure of Fig. \ref{fig:bayesian_net_general}.
We remark that (\ref{eq:rho_recursive_general}) is the same as with independent $X_1,\ldots,X_n$,
e.g., see equation (9) in \cite[Sec. IV]{MTR2013}.

\subsection{Computing The Marginal Probability}
Define $\X_m \equiv (X_{(m-1) L + 1},X_{(m-1) L + 2},\ldots,X_{(m-1) L + L})$.
Suppose the input \emph{symbols} are i.i.d. and
$\Xsymb{m} \in \mathcal{X}$ where $\mathcal{X}$ is a finite set.
Therefore, $p_{\Xsamp^n}$ has the form
\begin{align}
  p_{\Xsamp^n}(x^n)
  &= \prod_{m=1}^{\nsymb} p_{\X}(\x_m).
\end{align}
A Bayesian network for $\Xsamp^n,S^n,\Upsilon^n$ is shown in Fig. \ref{fig:bayesian_net_oversamp}.
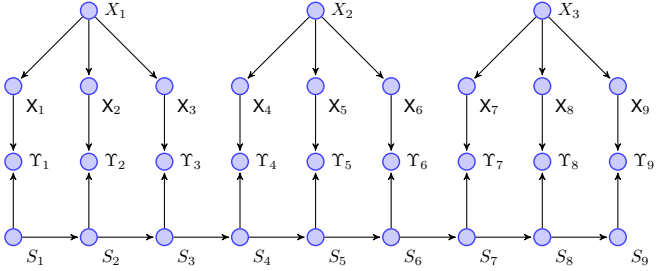
\begin{figure}
\centering
\resizebox{\ifdefined\twocolumnmode \columnwidth \else 0.75\textwidth \fi}{!}{\begin{tikzpicture}[->,>=stealth',shorten >=1pt,auto,node distance=1.5cm,semithick]
  \tikzstyle{var}=[circle,thick,draw=blue!75,fill=blue!20,minimum size=3pt]

  \node[var]  (S1) [		label=-45:$S_1$] {};
  \node[var]  (S2) [right of=S1,label=-45:$S_2$] {};
  \node[var]  (S3) [right of=S2,label=-45:$S_3$] {};
  \node[var]  (S4) [right of=S3,label=-45:$S_4$] {};
  \node[var]  (S5) [right of=S4,label=-45:$S_5$] {};
  \node[var]  (S6) [right of=S5,label=-45:$S_6$] {};
  \node[var]  (S7) [right of=S6,label=-45:$S_7$] {};
  \node[var]  (S8) [right of=S7,label=-45:$S_8$] {};
  \node[var]  (S9) [right of=S8,label=-45:$S_9$] {};

  \node[var]  (Y1) [above of=S1,label=right:$\Upsilon_1$] {};
  \node[var]  (Y2) [right of=Y1,label=right:$\Upsilon_2$] {};
  \node[var]  (Y3) [right of=Y2,label=right:$\Upsilon_3$] {};
  \node[var]  (Y4) [right of=Y3,label=right:$\Upsilon_4$] {};
  \node[var]  (Y5) [right of=Y4,label=right:$\Upsilon_5$] {};
  \node[var]  (Y6) [right of=Y5,label=right:$\Upsilon_6$] {};
  \node[var]  (Y7) [right of=Y6,label=right:$\Upsilon_7$] {};
  \node[var]  (Y8) [right of=Y7,label=right:$\Upsilon_8$] {};
  \node[var]  (Y9) [right of=Y8,label=right:$\Upsilon_9$] {};

  \node[var]  (X1) [above of=Y1,label=-45:$\Xsamp_1$] {};
  \node[var]  (X2) [right of=X1,label=-45:$\Xsamp_2$] {};
  \node[var]  (X3) [right of=X2,label=-45:$\Xsamp_3$] {};
  \node[var]  (X4) [right of=X3,label=-45:$\Xsamp_4$] {};
  \node[var]  (X5) [right of=X4,label=-45:$\Xsamp_5$] {};
  \node[var]  (X6) [right of=X5,label=-45:$\Xsamp_6$] {};
  \node[var]  (X7) [right of=X6,label=-45:$\Xsamp_7$] {};
  \node[var]  (X8) [right of=X7,label=-45:$\Xsamp_8$] {};
  \node[var]  (X9) [right of=X8,label=-45:$\Xsamp_9$] {};

  \node[var]  (M1) [above of=X2,label=right:$\Xsymb{1}$] {};
  \node[var]  (M2) [above of=X5,label=right:$\Xsymb{2}$] {};
  \node[var]  (M3) [above of=X8,label=right:$\Xsymb{3}$] {};

  \path 
  (S1) edge node {} (S2)
  (S2) edge node {} (S3)
  (S3) edge node {} (S4)
  (S4) edge node {} (S5)
  (S5) edge node {} (S6)
  (S6) edge node {} (S7)
  (S7) edge node {} (S8)
  (S8) edge node {} (S9)
  (X1) edge node {} (Y1)
  (X2) edge node {} (Y2)
  (X3) edge node {} (Y3)
  (X4) edge node {} (Y4)
  (X5) edge node {} (Y5)
  (X6) edge node {} (Y6)
  (X7) edge node {} (Y7)
  (X8) edge node {} (Y8)
  (X9) edge node {} (Y9)
  (S1) edge node {} (Y1)
  (S2) edge node {} (Y2)
  (S3) edge node {} (Y3)
  (S4) edge node {} (Y4)
  (S5) edge node {} (Y5)
  (S6) edge node {} (Y6)
  (S7) edge node {} (Y7)
  (S8) edge node {} (Y8)
  (S9) edge node {} (Y9)
  (M1) edge node {} (X1)
  (M1) edge node {} (X2)
  (M1) edge node {} (X3)
  (M2) edge node {} (X4)
  (M2) edge node {} (X5)
  (M2) edge node {} (X6)
  (M3) edge node {} (X7)
  (M3) edge node {} (X8)
  (M3) edge node {} (X9)
;
\end{tikzpicture}}
\caption{Bayesian network for ${\Xsamp^n,S^n,\Upsilon^n}$ for $n=9$ and $L=3$.}
\label{fig:bayesian_net_oversamp}
\end{figure}
The probability $p_{\Upsilon^n}(y^n)$ can be computed using
\begin{align}
  p_{\Upsilon^n}(y^n) 
  = \sum_{s \in \mathcal{S}} \psi_{\nsymb}(s) 
\end{align}
where $\psi_{m}(s) \equiv p_{S_{m L},\Y^m}(s,\y^m) $
which can be computed using the recursion:
\begin{align}
  &\psi_{m}(s) \label{eq:psi_recursive} \\
  &=  \sum_{\tilde{\x} \in \mathcal{X}_L} p_{\X}(\tilde{\x}) \
      \sum_{\tilde{s} \in \mathcal{S}} \psi_{m-1}(\tilde{s}) \
      p_{S_{m L},\Y_m|S_{(m-1)L},\X_m}(s,\y_m|\tilde{s},\tilde{\x})
  \nonumber
\end{align}
with the initial value 
$\psi_{0}(s) = 1/|\mathcal{S}|$.
The set $\mathcal{X}_L$ is
\begin{align}
 \mathcal{X}_L = \{ x \cdot (g(\Delta),g(2\Delta),\ldots,g(L\Delta) ): x \in \mathcal{X} \}.
\end{align}
We remark that $|\mathcal{X}_L| = |\mathcal{X}|$ and not $|\mathcal{X}|^L$.
Next, we define
\begin{align}
  \chi_{m,L}(s,\tilde{s},\tilde{\x}) = p_{S_{m L},\Y_m|S_{(m-1)L},\X_m}(s,\y_m|\tilde{s},\tilde{\x}) 
  \label{eq:chi_def}
\end{align}
for $s,\tilde{s} \in \mathcal{S}$ and $\tilde{\x} \in \mathcal{X}_L$.
Computing $\chi_{m,L}(s,\tilde{s},\tilde{\x})$ is similar
to computing $\rho_n$ (see (\ref{eq:rho_recursive_general})).
Intuitively, this is because a block of $L$ samples in Fig. \ref{fig:bayesian_net_oversamp} has a structure similar to Fig. \ref{fig:bayesian_net_general}.
More precisely, 
$\chi_{m,L}(s,\tilde{s},\tilde{\x})$ can be computed recursively by using
\begin{align}
  &\chi_{m,\ell}(s,\tilde{s},\tilde{\x}) \label{eq:chi_recursive} \\
  &= \sum_{ \varsigma \in \mathcal{S} } 
     \chi_{m,\ell-1}(\varsigma,\tilde{s},\tilde{\x}) \ 
     Q\left( s| \varsigma \right) \
     p_{\Psi|X,\Theta} \left( y_{(m-1)L+\ell}|\tilde{x}_{\ell}, s\right)
     \nonumber
\end{align}
with the initial value
\begin{align}
  \chi_{m,0}(s,\tilde{s},\tilde{\x}) = \left\{
  \begin{array}{ll}
   1, & s=\tilde{s} \\
   0, & \text{otherwise}.
  \end{array}
  \right.
  \label{eq:chi_init}
\end{align}

Therefore, computing $p_{\Upsilon^n}(y^n)$ involves two levels of recursion:
1) a recursion over the symbols as described by (\ref{eq:psi_recursive}) and 
2) a recursion over the samples within a symbol as described by (\ref{eq:chi_recursive}).

\section{Numerical Simulations}
\label{sec:num-sim}
We use two pulses with a symbol-interval time support:
\begin{itemize}
\item 
A unit-energy square pulse
\begin{align}
g_1(t) = 
\frac{1}{\sqrt{\Ts}} \text{rect}\left(\frac{t}{\Ts}\right)
\end{align}
where 
\begin{align}
 \text{rect}(t) \equiv \left\{ 
  \begin{array}{ll}
  1,	& |t| \leq 1/2 , \\
  0,	& \text{otherwise}.
  \end{array}
 \right.
\label{eq:rect_def}
\end{align}

\item 
A unit-energy cosine-squared pulse
\begin{align}
g_2(t) = 
\frac{1}{\sqrt{\Ts/2}} \cos^2\left(\frac{\pi t}{\Ts}\right) \text{rect}\left(\frac{t}{\Ts}\right).
\end{align}
\end{itemize}

The first step of the algorithm is to sample a long sequence according to the true joint distribution of $\Xsamp^n$ and $Y^n$. 
To generate samples according to the original channel (\ref{eq:Yk}),
we must accurately represent digitally the continuous-time waveform (\ref{eq:rx_waveform}).
We use a simulation oversampling rate $L_{\text{sim}}$ = 1024 samples/symbol.
After the filter (\ref{eq:yt}), the receiver has $L$ samples/symbol distributed according to (\ref{eq:Yk}).
Next, to choose a proper sequence length, we follow the approach suggested in \cite{Arnold2006}:
for a candidate length, run the algorithm about 10 times (each with a new random seed)
and check whether all estimates of the information rate agree up to the desired accuracy.
We used $\nsymb = 10^4$ unless otherwise stated.

For efficient implementation of (\ref{eq:rho_recursive_general}), 
$p_{\Psi|\Xsamp,\Theta}(\cdot|\cdot,\cdot)$  
can be factored out of the summation to yield:
\begin{align}
  \rho_{k}(s) 
  &=  p_{\Psi|\Xsamp,\Theta}(y_k|x_k,s)
      \overbrace{\sum_{\tilde{s} \in \mathcal{S}} \rho_{k-1}(\tilde{s}) \ Q(s|\tilde{s})}^{\rho^\prime_{k}(s)}.
  \label{eq:rho_recursive_general_implement}
\end{align}
Moreover, 
since $Q(\cdot|\cdot)$ can be represented by a circulant matrix due to symmetry,
$\rho^\prime_k(\cdot)$ can be computed efficiently using the Fast Fourier Transform (FFT).
Similarly, the computation of (\ref{eq:chi_recursive}) can be done efficiently  
by factoring out $p_{\Psi|X,\Theta}(\cdot|\cdot,\cdot)$ and
by using the FFT.

\subsection{Very Large Linewidth}

\newcommand{\SimFigPath}{./figures-uber-newlegend/}			
\begin{figure}[t]
\centering
\includegraphics[width= {\ifdefined\twocolumnmode \columnwidth \else 0.66\textwidth \fi}]{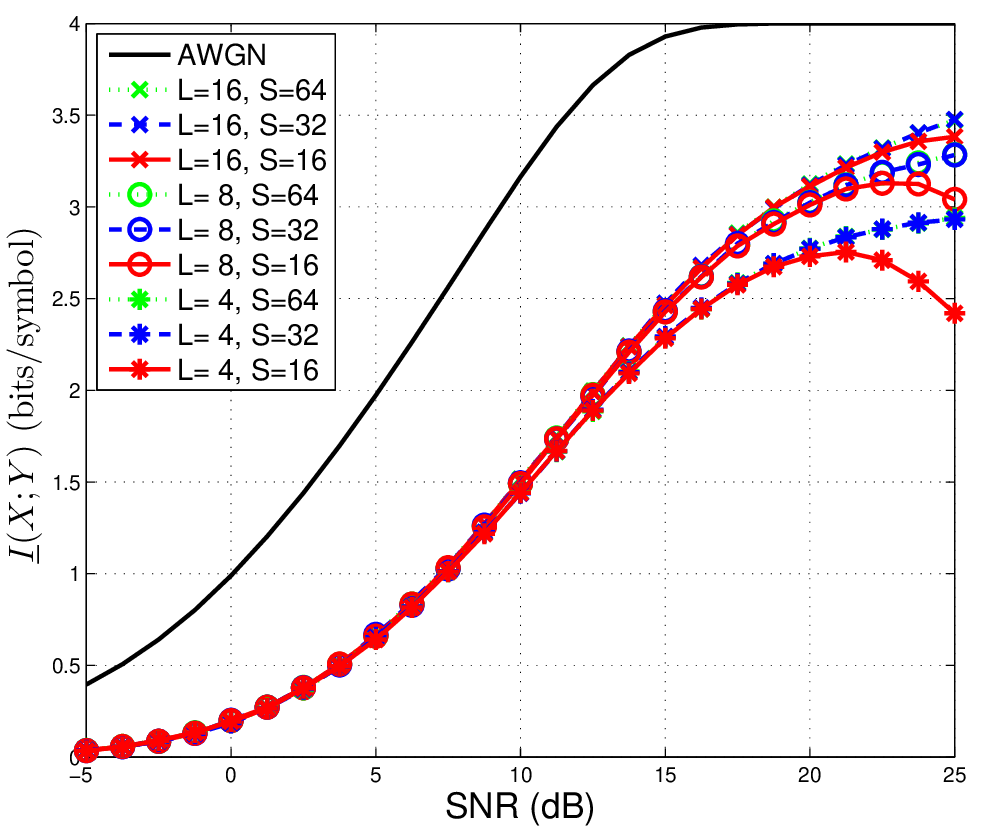}
\caption{Lower bounds on rates for 16-QAM, square transmit-pulse and multi-sample receiver at $f_{\text{HWHM}} \Ts = 0.125$.}
\label{fig:16qam_fhwhm125_sqrtx_optrx_osr}
\end{figure}
Suppose $f_{\text{HWHM}} \Ts = 0.125$ and the input symbols are independently and uniformly distributed (i.u.d.) 16-QAM.
Fig. \ref{fig:16qam_fhwhm125_sqrtx_optrx_osr} shows an estimate of $\underline{I}(X;Y)$ 
for a square transmit-pulse, 
i.e., $g(t) = g_1(t-\Ts/2)$
and an $L$-sample receiver  
with $L=4,8,16$ and $S=16,32,64$.
The curves with $S=64$ are indistinguishable from the curves with $S=32$ over the entire SNR range for all values of $L$,
and hence it seems that $S=32$ is adequate up to 25 dB.
Even $S=16$ seems to be adequate up to 20 dB.
We remark that the non-monotonicity of some lower bounds is an artifact of the auxiliary channel being ``far'' from the true channel,
specifically because the number $S$ of quantization levels of Auxiliary Channel II is too low for some SNR values.
The important trend in Fig. \ref{fig:16qam_fhwhm125_sqrtx_optrx_osr}
is that higher oversampling rate $L$ is needed at high $\SNR$
to extract all the information from the received signal.
For example,
$L=4$ suffices up to $\SNR$ $\sim$ 10 dB,
$L=8$ suffices up to $\SNR$ $\sim$ 15 dB but
$L \geq 16$ is needed beyond that.
It was pointed out in \cite{Arnold2006} that the lower bounds can be interpreted as 
the information rates achieved by mismatched decoding.
For example, $\underline{I}(X;Y)$ for $L=8$ and $S \geq 32$ in Fig. \ref{fig:16qam_fhwhm125_sqrtx_optrx_osr} is essentially
the information rate achieved by a multi-sample (8-sample) receiver
that uses maximum-likelihood decoding for the simplified channel (\ref{eq:Psik_def})
when it is operated in the original channel (\ref{eq:Yk}).

\begin{figure}[t]
\centering
\includegraphics[width= {\ifdefined\twocolumnmode \columnwidth \else 0.66\textwidth \fi}]{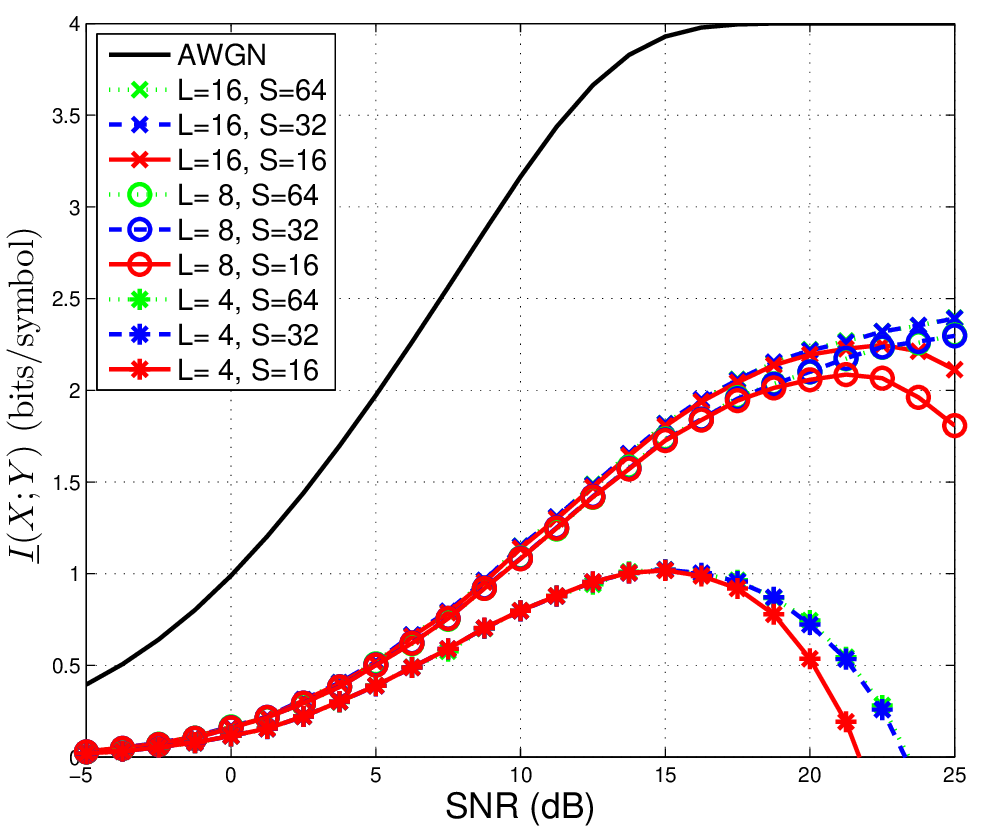}
\caption{Lower bounds on rates for 16-QAM, cosine-squared transmit-pulse and multi-sample receiver at $f_{\text{HWHM}} \Ts = 0.125$.}
\label{fig:16qam_fhwhm125_cos2tx_optrx_osr}
\end{figure}

Fig. \ref{fig:16qam_fhwhm125_cos2tx_optrx_osr} shows an estimate of $\underline{I}(X;Y)$ 
for a cosine-squared transmit-pulse, 
i.e., $g(t) = g_2(t-\Ts/2)$
and an $L$-sample receiver  
at $L=4,8,16$ and $S=16,32,64$.
We find that $S = 32$ suffices up to $\sim$ 25 dB.
We see in Fig. \ref{fig:16qam_fhwhm125_cos2tx_optrx_osr}
the same trend as in Fig. \ref{fig:16qam_fhwhm125_sqrtx_optrx_osr}:
higher $L$ is needed at higher $\SNR$.
Furthermore,
the square pulse is better than the cosine-squared pulse for the same oversampling rate $L$.
It might seem strange the the cosine-squared pulse curves are substantially lower
than the square pulse curves despite taking into account the pulse shape in the auxiliary
channel model, see (\ref{eq:Upsilonk_def}) and (\ref{eq:Xk_def}). 
One obtains insight on this effect with the following Gedankenexperiment. 
Suppose we use a pulse that is narrow and peaky in time. 
Then an integrate-and-dump receiver with oversampling puts out 
essentially only one useful sample per symbol. 
One may as well, therefore, use only one sample per symbol.

\subsection{Large Linewidth}

\begin{figure}[t]
\centering
\includegraphics[width= {\ifdefined\twocolumnmode \columnwidth \else 0.66\textwidth \fi}]{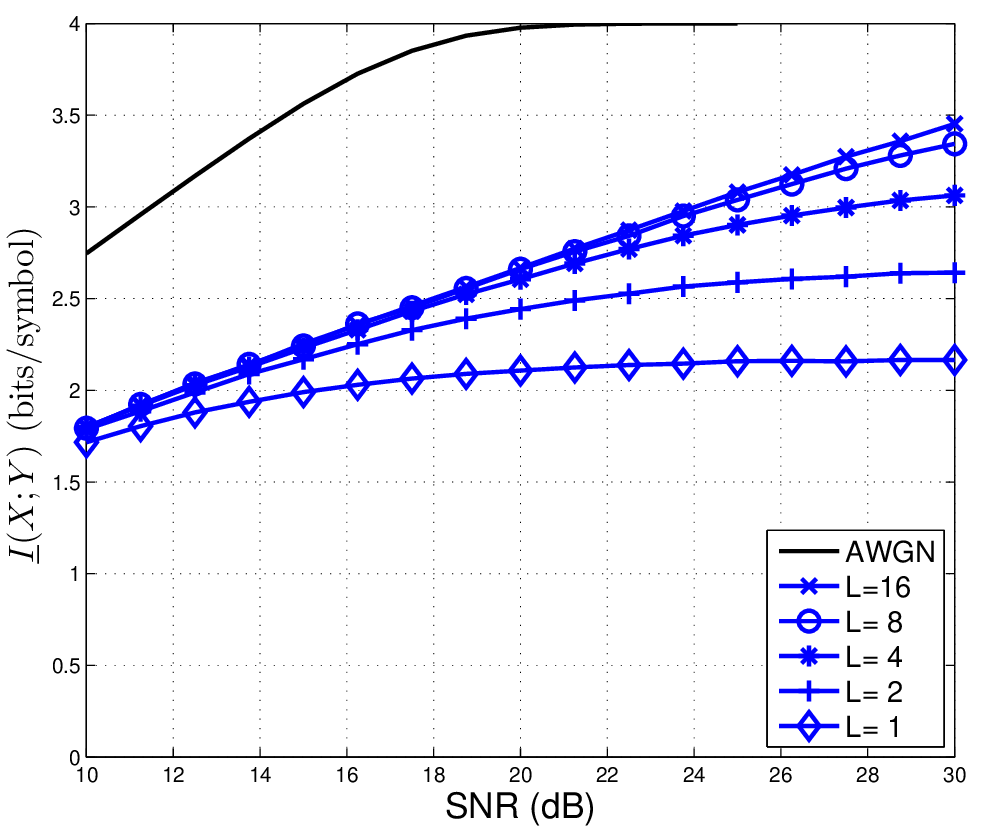}
\caption{Lower bounds on rates for 16-PSK, square transmit-pulse and multi-sample receiver at $f_{\text{HWHM}} \Ts = 0.0125$.}
\label{fig:16psk_fhwhm12_sqrtx_optrx_osr}
\end{figure}


As the linewidth decreases, the benefit of oversampling at the receiver becomes apparent only at higher $\SNR$.
For example, for $f_{\text{HWHM}} \Ts = 0.0125$ and i.u.d. 16-PSK input,
Fig. \ref{fig:16psk_fhwhm12_sqrtx_optrx_osr} shows an estimate of $\underline{I}(X;Y)$ 
for a square transmit-pulse 
and an $L$-sample receiver  
at $L=1,2,4,8,16$ and $S=64$.
We see that
$L=4$ suffices up to $\SNR$ $\sim$ 19 dB,
$L=8$ suffices up to $\SNR$ $\sim$ 24 dB and
only beyond that $L \geq 16$ is necessary.

We conclude from Fig. \ref{fig:16qam_fhwhm125_sqrtx_optrx_osr}--\ref{fig:16psk_fhwhm12_sqrtx_optrx_osr} that
the required $L$ depends on
1) the linewidth $f_{\text{FWHM}}$ of the phase noise;
2) the pulse $g(t)$; and
3) the $\SNR$.

\subsection{Comparison With Other Models}
We compare the discrete-time model of the multi-sample receiver
with other discrete-time models.
The simulation parameters for our model (labeled GK) are $\nsymb = 10^4$, $L = 16$ (with $L_{\text{sim}} = 1024$) and 
$S=64$ for 16-QAM ($S=128$ was too computationally intensive) and $S=128$ for QPSK.
We set $\gamma^2 = 2 \pi \beta \Ts$.
In Fig. \ref{fig:raheli_comparison}, we show curves
for the approximate symbol-rate model (\ref{eq:dt-mfrx-approx}), which are labeled ASM.
The simulation parameters for this model are  $\nsymb = 10^5$ and $S=128$.

\begin{figure}[t]
\centering
\includegraphics[width= {\ifdefined\twocolumnmode \columnwidth \else 0.66\textwidth \fi}]{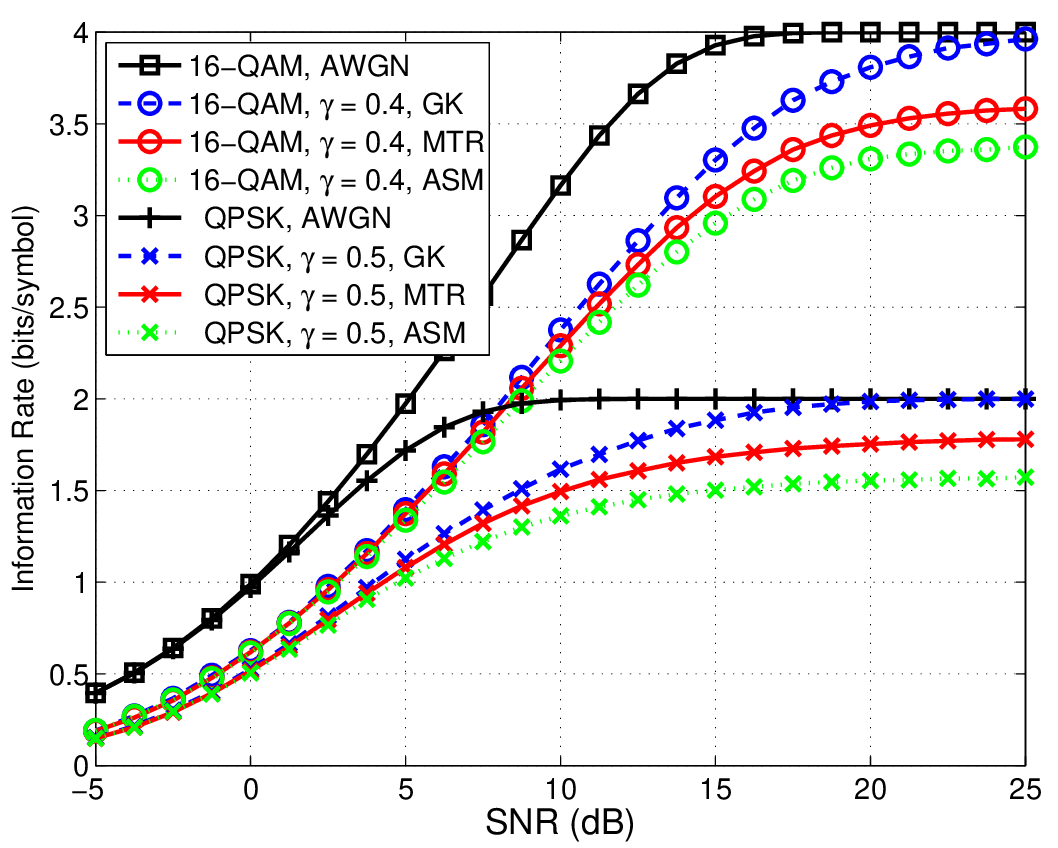}
\caption{Comparison of information rates for different models.}
\label{fig:raheli_comparison}
\end{figure}

We also show curves for the Martal\`{o}-Tripodi-Raheli (MTR) model \cite{MTR2013} in Fig. \ref{fig:raheli_comparison}.
For the sake of comparison, we adapt the model in \cite{MTR2013} 
from a square-root raised-cosine pulse to a square pulse 
and write the ``matched'' filter output $\{V_m\}$ as
\begin{align}
 V_m = \sum_{\ell = 1}^{L_{\text{sim}}} \Psi_{(m-1)L+\ell}
\end{align}
where $m=1,\ldots,\nsymb$ and $\Psi_k$ is defined in (\ref{eq:Psik_def}).
The auxiliary channel is 
\begin{align}
 Y_m = \Xsymb{m} ~ e^{j \Theta_m} + Z_m,
 \qquad m \geq 1
\end{align}
where 
the process $\{Z_m\}$ is an i.i.d. circularly-symmetric complex Gaussian process with mean $0$ and 
$\mathbb{E}[ |Z_m|^2 ] = \sigma^2_N \Ts$
while the process $\{\Theta_m\}$ is a first-order Markov process (not a Wiener process)
with a time-invariant transition probability, 
i.e., for $k \geq 2$ and all $\theta_k,\theta_{k-1} \in [-\pi,\pi)$, we have
$p_{\Theta_k|\Theta_{k-1}}(\theta_k|\theta_{k-1}) = p_{\Theta_2|\Theta_1}(\theta_k|\theta_{k-1})$.
Furthermore, the phase space is quantized to a finite number $S$ of states
and the transition probabilities are estimated by means of simulation.
The probabilities are then used to compute a lower bound on the information rate.
The simulation parameters for the MTR model are  $\nsymb = 10^5$, $L_{\text{sim}}=16$ and $S=128$.

We point out that MTR rates are essentially achievable rates for the symbol-rate model with filtering (\ref{eq:dt-mfrx-full}) 
that are valid up to some SNR that grows with the oversampling rate.
We observe from Fig. \ref{fig:raheli_comparison} the following:
\begin{enumerate}
\item 
The rates of ASM and MTR saturate well below
the rate achieved by the multi-sample receiver, 
which implies that symbol-rate sampling results in information loss at high SNR.

\item
The multi-sample receiver rates are indistinguishable from the MTR rates at low SNR which
implies that matched filtering with symbol-rate sampling incurs practically no information loss, i.e.,
using (\ref{eq:dt-mfrx-approx}) to compute rates at low/moderate SNR gives accurate results.

\item
MTR is indistinguishable from ASM at low SNR  
which implies that approximating (\ref{eq:dt-mfrx-full}) as (\ref{eq:dt-mfrx-approx}) can be justified at low SNR, i.e.,
the filtering effect may be ignored.
\end{enumerate}

We remark that the saturation at high SNR of the rates of the symbol-rate models at a value lower than the entropy of the source should not be surprising.
The impact of the additive noise is smaller at high SNR, but the effect of phase noise and/or filtering are not necessarily reduced.
For example, consider the memoryless channel $Y = H X + Z$ and  a fixed-size constellation (e.g. $M$-QAM). 
We have the following upper bound (by using a genie-aided argument)
\begin{align}
I(X;Y) &
\leq I(X;Y Z) \nonumber \\&
= \H(X) - \H(X|H X).
\end{align}
The term $\H(X|H X)$ represents impact of the ambiguity due to the randomness of $H$.
For the special case of $Y = e^{j\Theta} X + Z$, we have
\begin{align}
I(X;Y) \leq \H(X) - \H(\Phi_X|\Phi_X+\Theta)
\end{align}
i.e., one may fully recover the amplitude of $X$ but not the phase of $X$ because of the ambiguity due to $\Theta$.
Moreover, for $Y = e^{j\Theta} X + Z$, there is a more general upper bound on the contribution of phase modulation given by
\begin{align}
I(\Phi_X;Y|X_A) & 
\leq I(\Phi_X;Y|X_A,Z) \nonumber \\&
= I(\Phi_X;\Phi_X+\Theta|X_A) \nonumber \\&
= h(\Phi_X+\Theta|X_A) - h(\Theta) \nonumber \\&
\leq \log(2\pi) - h(\Theta)
\end{align}
This upper bound is more general than the earlier one in the sense that it is valid for any input constellation (of any size).
Unfortunately, we do not have a simple upper bound expression similar to $\log(2\pi) - h(\Theta)$ to predict the cap on the information rate for the multi-sample receiver.
Obtaining such an upper bound is beyond the scope of this paper.


\section{Discussion}
\label{sec:discuss}

\begin{enumerate}[wide, labelwidth=!, labelindent=0pt, itemsep=8pt]
\item
The authors of \cite{Foschini1988Comm} treated on-off keying in the presence of Wiener phase noise
by using a double-filtering receiver, which is composed of
an intermediate frequency (IF) filter, followed by an envelope detector (square-law device) and then a post-detection filter.
They showed that by optimizing the IF receiver bandwidth the double-filtering receiver outperforms the single-filtering (matched filter) receiver.
Furthermore, they showed via simulation that the optimum IF bandwidth increases with the SNR.
This is similar to our result in the sense that
we require the number of samples per symbol to increase with the SNR 
in order to achieve a rate that grows logarithmically with the SNR.
We remark that
Dallal and Shamai \cite{ShamaiCom91}
derived an upper bound on the error probability of (a variant of) the double filtering receiver 
for systems employing noncoherent demodulation, such as frequency shift keying (FSK), on-off keying (OOK) and pulse-position modulation (PPM).
Relevant analysis can be found in \cite{ShamaiIT92}, \cite{Shamai93} and references therein.
The analysis typically involves moments of filtered Wiener phase noise.

\item
The phase modulation pre-log of 1/4 in (\ref{eq:I_PhiX_Y|XA_approxmodel_theorem}) requires only 2 samples per symbol 
for which the time resolution $1/\Delta$ grows as the square root of the SNR. 
In other words, 
one does not need an analog-to-digital converter (ADC) with sampling every $T_s/L$ seconds.
Instead, one may use two ADCs that sample every $T_s$ seconds and the sampling instants are offset by $T_s/L$ seconds. 
For large $L$, two ADCs with low sampling rate would be less complex than one ADC with a very high sampling rate.
It is interesting to contemplate whether another receiver, e.g., a non-coherent receiver, 
can achieve the maximum amplitude modulation pre-log of 1/2 but requires only 1 sample per symbol. 
If so, one would need only 3 samples per symbol to achieve a pre-log of 3/4.


\item
The paper \cite{GhozlanISIT2013} shows that a pre-log of 1/2 is achievable 
for the model (\ref{eq:Yk}) when $L \propto \sqrt{\SNR}$.
Theorem \ref{theorem:oversamp-fullmodel} implies that a pre-log of 1/2 can be achieved using $L \propto \sqrt[3]{\SNR}$
and this perhaps can be reduced further.
But using $L \propto \sqrt{\SNR}$ may simplify the receiver design.
Specifically, the amplitude decoder designed for the auxiliary channel in \cite{GhozlanISIT2013} can be used for any phase noise linewidth $\beta$ 
because the auxiliary channel is (\ref{eq:Q_V|XA}) with $\tilde{G} = 1$ which is independent of $\beta$.
In contrast, using $L = \sqrt[3]{\beta^2 \SNR}$ requires different amplitude decoders for different linewidths
because the auxiliary channel in this case is (\ref{eq:Q_V|XA}) with $\tilde{G}=\mathbb{E}[G]$ which depends on $\beta$.

\item
For the multi-sample approximate model (\ref{eq:Yk_approx}), Barletta and Kramer \cite{BarlettaITW2015} derived an upper bound on the rate and showed that if
\begin{align}
 L = \left\lceil \SNR^\alpha \right\rceil
\end{align}
where $0 < \alpha \leq 1$, then the capacity pre-log is \emph{at most} $(1+\alpha)/2$
which implies that $L \propto \sqrt{\SNR}$ is optimal for achieving a pre-log of 3/4.
We remark that the best known oversampling factor to achieve a pre-log of 3/4 for the model (\ref{eq:Yk}) is $L \propto \sqrt{\SNR}$.
We also remark that the upper bound derived in \cite{BarlettaITW2015} was subsequently specialized to the additive-noise-free case in \cite{MTR2016} (cf. (141)-(143)).

\item
We briefly discuss sufficient statistics for the continuous-time phase noise channel.
One expects intuitively at any SNR that
\begin{align}
	\lim_{L \rightarrow \infty} I(X^\nsymb;Y^{\nsymb L}) = I(x(t); r(t)).
\end{align}
This is because at infinite sampling rate, one captures essentially the entire received waveform.
In general, it seems that an infinite number of projections per symbol are needed 
to obtain sufficient statistics for all SNRs.
For example, consider 
the Karhunen-Lo\`{e}ve expansion of $U(t) = e^{j \Theta(t)}$ on $t \in [0,\Ts]$, for which we have\footnote{
We use $\stackrel{\text{m.s.}}{=}$ to indicate convergence in mean-square sense.}
\begin{align}
	U(t) \stackrel{\text{m.s.}}{=} \sum_{i=0}^\infty U_i f_i(t)
\end{align}
where $\{f_i(t)\}$ are the orthonormal functions \cite[Sec. 6.1]{AzizogluPHD1991}:
\begin{align}
	 f_i(t) = \sqrt{\frac{2}{\Ts + \lambda_i}}
   \left\{
	 \begin{array}{l}
			\cos(\alpha_i(t-\Ts/2)), ~ i \text{ odd} \\
			\sin(\alpha_i(t-\Ts/2)), ~ i \text{ even}
	 \end{array}
	\right.
\end{align} 
where $\alpha_i$ is the solution to
\begin{align}
2 \arctan(\alpha_i/(\beta\pi)) + \alpha_i \Ts = i \pi, \quad i=0,1,\ldots
\end{align}
and the $\{U_i\}$ are uncorrelated random variables with variance $\lambda_i$
where
\begin{align}
\lambda_i = \frac{2\pi\beta}{\pi^2 \beta^2+\alpha_i^2}.
\end{align}
Suppose that one symbol $X$ is transmitted and consider a rectangular pulse for simplicity. We have
\begin{align}
	r(t) = \Xsymbol ~ e^{j \Theta(t)} + n(t) = \sum_{i=0}^\infty \Xsymbol U_i ~ f_i(t) + n(t).
\end{align}
We conclude, in this simplified case, that
\begin{align}
	I(\Xsymbol;\{\Gamma_i\}_{i=0}^\infty) = I(x(t); r(t))
\end{align}
where
\begin{align}
\Gamma_i = \int_0^{\Ts} r(t) f_i^*(t) dt, \qquad i=0,1,\ldots
\label{eq:dt-kl-expansion}
\end{align}
i.e, 
$\{\Gamma_i\}_{i=0}^\infty$ are sufficient statistics for $\Xsymbol$.
Therefore, an infinite number of projections per symbol seems to be needed 
to have sufficient statistics in general.
Moreover, we believe that a finite subset of $\{\Gamma_i\}_{i=0}^\infty$ 
captures essentially all the information in the received waveform at a given SNR, 
and that larger subsets would be needed at larger SNRs.
The discrete-time model based on the multi-sample receiver in Fig. \ref{fig:rx_oversamp} seems to be easier to analyze than the discrete-time model (\ref{eq:dt-kl-expansion}).

\item
Minimum energy per information bit is an important metric at low SNR.
For example, Fig. \ref{fig:ebno} plots the information rates of Fig. \ref{fig:16qam_fhwhm125_sqrtx_optrx_osr} against $E_b/N_0$, 
where the energy per information bit is $E_b = P/R$ and $R$ is the information rate in bits/symbol.
\begin{figure}
   \centering
	 \includegraphics[width= {\ifdefined\twocolumnmode \columnwidth \else 0.66\textwidth \fi}]{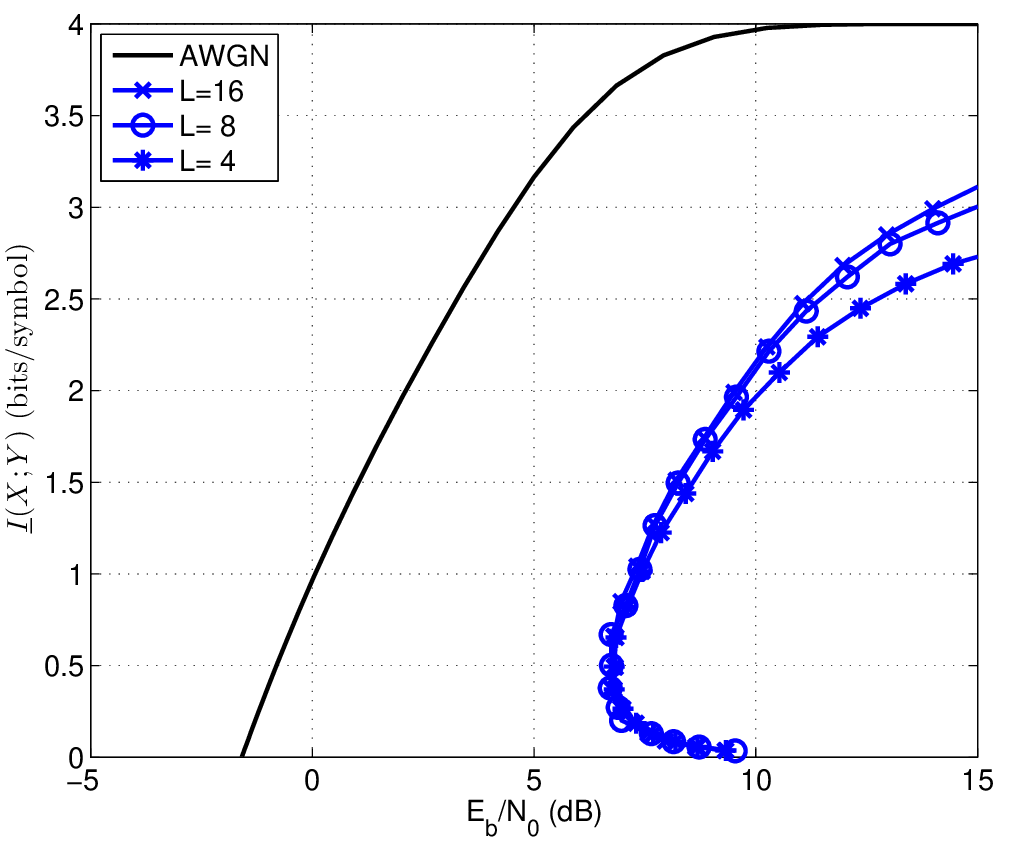}
	 \caption{Rates vs $E_b/N_0$ for 16-QAM, square transmit-pulse and multi-sample receiver at $f_{\text{HWHM}} \Ts = 0.125$.}
	 \label{fig:ebno}
\end{figure}
We notice a significant impact on the minimum $E_b$ due to phase noise.
Furthermore, oversampling does not affect the rate at low $E_b/N_0$,
which suggests that using symbol-rate sampling suffices 
and that several results available in the literature on noncoherent fading channels 
and discrete-time phase noise channels (with the commonly-used approximation)
should remain valid, e.g., see Fig. 2--Fig. 5 in \cite{Barbieri2011}.

\end{enumerate}


\section{Conclusion}
\label{sec:conc}
We studied a waveform channel impaired by Wiener phase noise and AWGN.
A discrete-time channel model based on oversampling was derived taking into account filtering effects on phase noise.
At high SNR, the multi-sample receiver achieves rates that grow logarithmically with SNR
with at least a 1/2 pre-log factor,
if the number of samples per symbol grows with the cubic root of the SNR.
In addition, we computed via numerical simulations lower bounds on the information rates achieved by the multi-sample receiver.
We observed that the required oversampling rate depends on
the linewidth of the phase noise, the shape of the transmit-pulse and the SNR.
The results demonstrate that multi-sample receivers increase the information rate
for both strong and weak phase noise at high SNR.
We compared our results with the results obtained by using other discrete-time models.
Finally, we showed for an approximate discrete-time model of the multi-sample receiver
that phase modulation achieves a pre-log of at least 1/4 while amplitude modulation achieves a pre-log of 1/2,
if the number of samples per symbol grows with the square root of the SNR. 
Thus, we demonstrated that the overall capacity pre-log is $3/4$ which is
greater than the capacity pre-log of the (approximate) discrete-time Wiener phase noise channel with only one sample per symbol.
The impact of oversampling on the capacity of 
non-Wiener phase noise channels, or multi-input multi-output (MIMO) phase noise channels,
are open problems.

\section*{Acknowledgement}
We thank the reviewers and the Associate Editor for useful comments.

\appendices 


\section{Amplitude Modulation}
\label{sec:amplitude-lower-bound}
This Appendix provides a proof of Theorem \ref{theorem:oversamp-fullmodel}.
We set $\Ts=1$ for simplicity.
Since $X_{A}^\nsymb$ is i.i.d.,
we have
\begin{align}
I(X_{A}^\nsymb;\Y^\nsymb)
&\stackrel{(a)}{=} \sum_{k=1}^\nsymb I(X_{A,k};\Y^\nsymb|X_{A}^{k-1}) \nonumber \\
&\stackrel{(b)}{=} \sum_{k=1}^\nsymb \H(X_{A,k}) - \H(X_{A,k}|\Y^\nsymb ~ X_{A}^{k-1}) \nonumber \\
&\stackrel{(c)}{\geq} \sum_{k=1}^\nsymb I(X_{A,k};\Y_k) \nonumber \\
&\stackrel{(d)}{\geq} \sum_{k=1}^\nsymb I(X_{A,k};V_k)
\label{eq:I_XA1n_Y1n_LB}
\end{align}
where
\begin{align}
V_k = \sum_{\ell=1}^L |Y_{(k-1)L+\ell}|^2.
\label{eq:V_def}
\end{align}
Step
$(a)$ follows from the chain rule of mutual information, 
$(b)$ follows from the independence of $X_{A,1},X_{A,2},\ldots,X_{A,b}$, 
$(c)$ holds because conditioning does not increase entropy, and
$(d)$ follows from the data processing inequality.
Since $X_{A}^\nsymb$ is identically distributed, we find that $V^\nsymb$ is also identically distributed and 
we have, for $k \geq 2$,
\begin{align}
I(X_{A,k};V_k) = I(X_{A,1};V_1).
\label{eq:I_XAk_Vk}
\end{align}

In the rest of this section, we consider only one symbol ($k=1$) and drop the time index.
%
By combining (\ref{eq:V_def}) and (\ref{eq:Yk}), we have
\begin{align}
V 
&= \sum_{\ell=1}^L \left( X_A^2 \Delta^2 |F_\ell|^2 + 2 X_A \Delta \Re[ e^{j \Phi_X} e^{j \Theta_\ell} F_\ell N_\ell^* ] + |N_\ell|^2 \right) \nonumber \\
&= X_A^2 \Delta G+ 2 X_A \Delta Z_1 + Z_0
\label{eq:V_direct}
\end{align}
where $G$, $Z_1$ and $Z_0$ are defined as
\begin{align}
G &\equiv \frac{1}{L} \sum_{\ell=1}^L |F_\ell|^2
\label{eq:G_def} \\
Z_1 &\equiv \sum_{\ell=1}^L \Re[ e^{j \Phi_X} e^{j \Theta_\ell} F_\ell N_\ell^* ]
\label{eq:Z1_def} \\
Z_0 &\equiv \sum_{\ell=1}^L |N_\ell|^2.
\label{eq:Z0_def} 
\end{align}
The second-order statistics of $Z_1$ and $Z_0$ are
\begin{align}
\begin{array}{ll}
\mathbb{E}[ Z_1 ] = 0 		& \textsf{Var}[ Z_1 ] = \mathbb{E}[G] {\sigma^2_N}/{2} \\
\mathbb{E}[ Z_0 ] = \sigma^2_N 	& \textsf{Var}\left[ Z_0 \right] = \sigma^4_N \Delta \\
\mathbb{E}\left[ Z_1 (Z_0-\mathbb{E}[Z_0]) \right] = 0.
\end{array}
\label{eq:second-order-stats}
\end{align} 


By using the auxiliary-channel lower bound (cf. (\ref{eq:I_AuxCh})), we have
\begin{align}
I(X_A;V)
&\geq \mathbb{E}[-\log{q_V(V)}] - \mathbb{E}[-\log{q_{V|X_A}(V|X_A)}]
\label{eq:I_XA_V_LB_AuxCh}
\end{align}
where $q_{V|X_A}(v|x_A)$ is the conditional distribution of the auxiliary channel,
which we choose to be
\begin{align}
q_{V|X_A}(v|x_A) 
= \frac{1}{\sqrt{4 \pi x_A^2 \Delta^2 \sigma^2_N}}
\exp\left(- \frac{(v-x_A^2 \Delta \tilde{G} - \sigma^2_N)^2}{4 x_A^2 \Delta^2 \sigma^2_N} \right)
\label{eq:Q_V|XA}
\end{align}
with $0<\tilde{G}\leq 1$ and 
\begin{align}
q_V(v) = \int_0^\infty p_{X_A}(x_A) q_{V|X_A}(v|x_A) dx_A
\label{eq:QV_def}
\end{align}
where $p_{X_A}(\cdot)$ is the (true) pdf of $X_A$.
The intuition behind the choice of (\ref{eq:Q_V|XA}) is given in Remark \ref{remark:amplitude-intuition} below.
It follows that
\ifdefined\twocolumnmode{
\begin{align}
&\mathbb{E}[ -\log( q_{V|X_A}(V|X_A) ) ]
= \mathbb{E}\left[ \frac{(V-X_A^2 \Delta \tilde{G} - \sigma^2_N)^2}{4 X_A^2 \Delta^2 \sigma^2_N} \right] \nonumber \\& \qquad \qquad
+ \log{\Delta} 
+ \frac{1}{2} \log(4 \pi\sigma^2_N)
+ \frac{1}{2} \mathbb{E}[ \log(X_A^2) ]
\label{eq:log_Q_V|XA_original}
\end{align}
}\else{
\begin{align}
\mathbb{E}[ -\log( q_{V|X_A}(V|X_A) ) ]
= \mathbb{E}\left[ \frac{(V-X_A^2 \Delta \tilde{G} - \sigma^2_N)^2}{4 X_A^2 \Delta^2 \sigma^2_N} \right]
+ \log{\Delta} 
+ \frac{1}{2} \log(4 \pi\sigma^2_N)
+ \frac{1}{2} \mathbb{E}[ \log(X_A^2) ]
\label{eq:log_Q_V|XA_original}
\end{align}
}\fi
and by using (\ref{eq:V_direct}) and (\ref{eq:second-order-stats}), we have
\ifdefined\twocolumnmode{
\begin{align}
&\mathbb{E}\left[ \frac{(V-X_A^2 \Delta \tilde{G} - \sigma^2_N)^2}{4 X_A^2 \Delta^2 \sigma^2_N} \right] \nonumber\\
&=\frac{1}{4 \sigma^2_N} \mathbb{E}[ X_A^2 ] \mathbb{E}\left[ (G-\tilde{G})^2 \right]
+ \frac{1}{  \sigma^2_N} \mathbb{E}[ Z_1^2 ] \nonumber\\&
+ \frac{1}{4 \Delta^2 \sigma^2_N} \mathbb{E}\left[ \frac{1}{X_A^2} \right] \mathbb{E}\left[ (Z_0-\sigma^2_N)^2 \right]
+ \frac{1}{  \sigma^2_N} \mathbb{E}[X_A] \mathbb{E}\left[ (G-\tilde{G}) Z_1 \right] \nonumber\\&
+ \frac{1}{2 \Delta \sigma^2_N} \mathbb{E}[G-1] \mathbb{E}[Z_0-\sigma^2_N] 
+ \frac{1}{  \Delta \sigma^2_N} \mathbb{E}\left[ \frac{1}{X_A}\right] \mathbb{E}\left[ Z_1 (Z_0-\sigma^2_N) \right] \nonumber \\
&=\frac{1}{4 \sigma^2_N} P \ \mathbb{E}\left[ (G-\tilde{G})^2 \right] 
+ \frac{1}{2} \mathbb{E}[G]
+ \frac{\sigma^2_N}{4 \Delta} \mathbb{E}\left[ \frac{1}{X_A^2} \right]
\label{eq:big_expectation}
\end{align}
}\else{
\begin{align}
\mathbb{E}\left[ \frac{(V-X_A^2 \Delta \tilde{G} - \sigma^2_N)^2}{4 X_A^2 \Delta^2 \sigma^2_N} \right] 
&=\frac{1}{4 \sigma^2_N} \mathbb{E}[ X_A^2 ] \mathbb{E}\left[ (G-\tilde{G})^2 \right]
+ \frac{1}{  \sigma^2_N} \mathbb{E}[ Z_1^2 ]
+ \frac{1}{4 \Delta^2 \sigma^2_N} \mathbb{E}\left[ \frac{1}{X_A^2} \right] \mathbb{E}\left[ (Z_0-\sigma^2_N)^2 \right] \nonumber\\&
+ \frac{1}{  \sigma^2_N} \mathbb{E}[X_A] \mathbb{E}\left[ (G-\tilde{G}) Z_1 \right]
+ \frac{1}{2 \Delta \sigma^2_N} \mathbb{E}[G-1] \mathbb{E}[Z_0-\sigma^2_N] 
+ \frac{1}{  \Delta \sigma^2_N} \mathbb{E}\left[ \frac{1}{X_A}\right] \mathbb{E}\left[ Z_1 (Z_0-\sigma^2_N) \right] \nonumber \\
&=\frac{1}{4 \sigma^2_N} P \ \mathbb{E}\left[ (G-\tilde{G})^2 \right] 
+ \frac{1}{2} \mathbb{E}[G]
+ \frac{\sigma^2_N}{4 \Delta} \mathbb{E}\left[ \frac{1}{X_A^2} \right]
\label{eq:big_expectation}
\end{align}
}\fi
where we also used
\begin{align}
 \mathbb{E}\left[ (G-\tilde{G}) Z_1 \right] = 0.
\end{align}
Substituting (\ref{eq:big_expectation}) into (\ref{eq:log_Q_V|XA_original}) and using $\mathbb{E}[G] \leq 1$ yield
\ifdefined\twocolumnmode{
\begin{align}
&\mathbb{E}[ -\log( q_{V|X_A}(V|X_A) ) ] \nonumber \\
&\leq \log{\Delta}
+ \frac{1}{2} \log(4 \pi\sigma^2_N)
+ \frac{1}{2} \mathbb{E}[ \log(X_A^2) ] \nonumber \\&
+ \frac{P}{4 \sigma^2_N} \mathbb{E}\left[ (G-\tilde{G})^2 \right]  + \frac{1}{2} + \frac{\sigma^2_N}{4 \Delta} \mathbb{E}\left[ \frac{1}{X_A^2} \right].
\label{eq:log_Q_V|XA}
\end{align}
}\else{
\begin{align}
\mathbb{E}[ -\log( q_{V|X_A}(V|X_A) ) ]
&\leq \log{\Delta}
+ \frac{1}{2} \log(4 \pi\sigma^2_N)
+ \frac{1}{2} \mathbb{E}[ \log(X_A^2) ] 
+ \frac{P}{4 \sigma^2_N} \mathbb{E}\left[ (G-\tilde{G})^2 \right]  + \frac{1}{2} + \frac{\sigma^2_N}{4 \Delta} \mathbb{E}\left[ \frac{1}{X_A^2} \right].
\label{eq:log_Q_V|XA}
\end{align}
}\fi

It is convenient to define $X_P \equiv X_A^2$.
The input distribution of $X_P$ is
\begin{align}
p_{X_P}(x_P) = \left\{ 
  \begin{array}{ll}
  \frac{1}{\lambda} \exp\left(-\frac{x_P-\lambda}{\lambda}\right),	& x_P \geq \lambda \\
  0,			& \text{otherwise}
  \end{array}
 \right.
\label{eq:P_XP}
\end{align}
where $\lambda = P/2$.
It follows from  (\ref{eq:QV_def}) and (\ref{eq:P_XP}) that
\begin{align}
q_V(v) 
&= \int_{\lambda}^\infty \frac{1}{\lambda} \exp\left(-\frac{x_P-\lambda}{\lambda}\right) ~ q_{V|X_P}(v|x_P) ~ dx_P \nonumber \\
&\leq e ~ f_{V}(v)
\label{eq:QV_FV_ineq}
\end{align}
where
\begin{align}
q_{V|X_P}(v|x_P) = q_{V|X_A}(v|\sqrt{x_P})
\label{eq:Q_V|XP}
\end{align}
and
\begin{align}
f_V(v)
&\equiv \int_{0}^\infty \frac{1}{\lambda} \exp\left(-\frac{x_P}{\lambda}\right) \ q_{V|X_P}(v|x_P) dx_P.
\label{eq:F_V_def}
\end{align}
The inequality (\ref{eq:QV_FV_ineq}) follows from the non-negativity of the integrand.
By combining (\ref{eq:Q_V|XA}), (\ref{eq:Q_V|XP}), (\ref{eq:F_V_def}) and making the change of variables $x=x_P \Delta$, we have
\ifdefined\twocolumnmode{
\begin{align}
&f_V(v) \nonumber\\
&= \int_0^\infty \frac{e^{-{x}/(\lambda \Delta)}}{\lambda \Delta}
\frac{1}{\sqrt{4 \pi x \Delta \sigma^2_N}} \exp\left(- \frac{(v-x - \sigma^2_N)^2}{4 x \Delta \sigma^2_N} \right) dx \nonumber \\
&= \frac{1}{\sqrt{\lambda \Delta (\lambda \Delta + 4 \Delta \sigma^2_N)}} ~\times \nonumber\\& \quad
\exp\left( \frac{2}{4 \Delta \sigma^2_N} \left[ v-\sigma^2_N - |v-\sigma^2_N| \sqrt{1+\frac{4 \Delta \sigma^2_N}{\lambda \Delta}} \right] \right)
\label{eq:F_V}
\end{align}
}\else{
\begin{align}
f_V(v)
&= \int_0^\infty \frac{e^{-{x}/(\lambda \Delta)}}{\lambda \Delta}
\frac{1}{\sqrt{4 \pi x \Delta \sigma^2_N}} \exp\left(- \frac{(v-x - \sigma^2_N)^2}{4 x \Delta \sigma^2_N} \right) dx \nonumber \\
&= \frac{1}{\sqrt{\lambda \Delta (\lambda \Delta + 4 \Delta \sigma^2_N)}} ~
\exp\left( \frac{2}{4 \Delta \sigma^2_N} \left[ v-\sigma^2_N - |v-\sigma^2_N| \sqrt{1+\frac{4 \Delta \sigma^2_N}{\lambda \Delta}} \right] \right)
\label{eq:F_V}
\end{align}
}\fi
where we used equation (140) in Appendix A of \cite{Moser2012}:
\begin{align}
&\int_0^\infty \frac{1}{a} \exp\left(-\frac{x}{a}\right)
\frac{1}{\sqrt{\pi b x}} \exp\left(- \frac{(u-x)^2}{b x} \right) dx \nonumber\\&
= \frac{1}{\sqrt{a (a + b)}}
\exp\left( \frac{2}{b} \left[ u - |u| \sqrt{1+\frac{b}{a}} \right] \right).
\end{align}
Therefore, we have
\ifdefined\twocolumnmode{
\begin{align}
&\mathbb{E}[ -\log( f_V(V) ) ] \nonumber\\
&= \frac{1}{2} \log(\Delta^2 (\lambda^2 + 4 \lambda \sigma^2_N)) \nonumber \\&
- \frac{1}{2 \Delta \sigma^2_N} \left[ \mathbb{E}[V-\sigma^2_N] - \mathbb{E}[|V-\sigma^2_N|] \sqrt{1+\frac{4 \sigma^2_N}{\lambda}} \right] \nonumber \\
&\stackrel{(a)}{\geq} \log(\Delta \lambda) + 
\frac{1}{2 \sigma^2_N \Delta} \mathbb{E}[V-\sigma^2_N] \left[ \sqrt{1+\frac{4 \sigma^2_N}{\lambda}} - 1 \right] \nonumber \\
&\stackrel{(b)}{\geq} \log(\Delta \lambda)
\label{eq:log_F_V}
\end{align}
}\else{
\begin{align}
\mathbb{E}[ -\log( f_V(V) ) ]
&= \frac{1}{2} \log(\Delta^2 (\lambda^2 + 4 \lambda \sigma^2_N))
- \frac{1}{2 \Delta \sigma^2_N} \left[ \mathbb{E}[V-\sigma^2_N] - \mathbb{E}[|V-\sigma^2_N|] \sqrt{1+\frac{4 \sigma^2_N}{\lambda}} \right] \nonumber \\
&\stackrel{(a)}{\geq} \log(\Delta \lambda) + 
\frac{1}{2 \sigma^2_N \Delta} \mathbb{E}[V-\sigma^2_N] \left[ \sqrt{1+\frac{4 \sigma^2_N}{\lambda}} - 1 \right] \nonumber \\
&\stackrel{(b)}{\geq} \log(\Delta \lambda)
\label{eq:log_F_V}
\end{align}
}\fi
where
$(a)$ holds because the logarithmic function is monotonic and $\mathbb{E}[| \cdot |] \geq \mathbb{E}[ \cdot ]$, and
$(b)$ holds because
\begin{align}
&\mathbb{E}[V-\sigma^2_N] \nonumber \\
&= \mathbb{E}[X_A^2] \Delta \mathbb{E}[G] + 2 \mathbb{E}[X_A] \Delta \mathbb{E}[Z_1] + \mathbb{E}[Z_0] - \sigma^2_N \nonumber \\
&= P \Delta \mathbb{E}[G] \geq 0.
\end{align}
The monotonicity of the logarithmic function and (\ref{eq:QV_FV_ineq}) yield
\begin{align}
\mathbb{E}[ -\log( q_{V}(V) ) ] 
&\geq \mathbb{E}\left[ -\log\left( e ~ f_V(V) \right) \right] \nonumber \\
&\geq \log{\Delta} + \log{\lambda}
- 1 
\label{eq:log_Q_V}
\end{align}
where the last inequality follows from (\ref{eq:log_F_V}).
It follows from (\ref{eq:I_XA_V_LB_AuxCh}), (\ref{eq:log_Q_V|XA}) and (\ref{eq:log_Q_V}) that
\ifdefined\twocolumnmode{
\begin{align}
I(X_A;V)
&\geq 
\log{\lambda}
- 1
- \frac{1}{2} \log(4 \pi\sigma^2_N)
- \frac{1}{2} \mathbb{E}[ \log(X_A^2) ] \nonumber \\&
- \frac{P}{4 \sigma^2_N} \mathbb{E}\left[ (G-\tilde{G})^2 \right] - \frac{1}{2} - \frac{\sigma^2_N}{4 \Delta} \mathbb{E}\left[ \frac{1}{X_A^2} \right].
\label{eq:I_XA_V_LB}
\end{align}
}\else{
\begin{align}
I(X_A;V)
&\geq 
\log{\lambda}
- 1
- \frac{1}{2} \log(4 \pi\sigma^2_N)
- \frac{1}{2} \mathbb{E}[ \log(X_A^2) ] 
- \frac{P}{4 \sigma^2_N} \mathbb{E}\left[ (G-\tilde{G})^2 \right] - \frac{1}{2} - \frac{\sigma^2_N}{4 \Delta} \mathbb{E}\left[ \frac{1}{X_A^2} \right].
\label{eq:I_XA_V_LB}
\end{align}
}\fi

We have 
\begin{align}
\mathbb{E}\left[ \frac{1}{X_P} \right] 
\leq \frac{2}{P}
\end{align}
and
\begin{align}
\mathbb{E}\left[ \log\left( X_P \right) \right]
&= \int_{\lambda}^\infty \frac{1}{\lambda} e^{-(x-\lambda)/\lambda} \log(x) dx \nonumber \\
&\stackrel{(a)}{=} \log\lambda + \int_{1}^\infty e^{-(u-1)} \log(u) du \nonumber \\
&\stackrel{(b)}{\leq} \log\lambda + 1
\end{align}
where 
$(a)$ follows by the change of variables $u = x/\lambda$, and
$(b)$ holds because $\log(u) \leq u-1$ for all $u > 0$.
Substituting into (\ref{eq:I_XA_V_LB}), we obtain
\ifdefined\twocolumnmode{
\begin{align}
I(X_A;V) - \frac{1}{2} \log{\SNR}
&\geq - 2 - \frac{1}{2} \log(8 \pi) 
- \frac{1}{2 \SNR \Delta} \nonumber\\&
- \frac{1}{4} \SNR ~ \mathbb{E}\left[ (G-\tilde{G})^2 \right].
\label{eq:I_XA_V_LB_AuxCh_final}
\end{align}
}\else{
\begin{align}
I(X_A;V) - \frac{1}{2} \log{\SNR}
&\geq - 2 - \frac{1}{2} \log(8 \pi) 
- \frac{1}{2 \SNR \Delta} 
- \frac{1}{4} \SNR ~ \mathbb{E}\left[ (G-\tilde{G})^2 \right].
\label{eq:I_XA_V_LB_AuxCh_final}
\end{align}
}\fi

By combining (\ref{eq:I_XA_Y-def}), (\ref{eq:I_XA1n_Y1n_LB}), (\ref{eq:I_XAk_Vk}) and (\ref{eq:I_XA_V_LB_AuxCh_final}), we have
\ifdefined\twocolumnmode{
\begin{align}
I(X_A;Y) - \frac{1}{2} \log{\SNR}
&\geq - 2 - \frac{1}{2} \log(8 \pi) 
- \frac{1}{2 \SNR \Delta} \nonumber\\&
- \frac{1}{4} \SNR ~ \mathbb{E}\left[ (G-\tilde{G})^2 \right].
\label{eq:I_XA_Y_LB}
\end{align}
}\else{
\begin{align}
I(X_A;Y) - \frac{1}{2} \log{\SNR}
&\geq - 2 - \frac{1}{2} \log(8 \pi) 
- \frac{1}{2 \SNR \Delta} 
- \frac{1}{4} \SNR ~ \mathbb{E}\left[ (G-\tilde{G})^2 \right].
\label{eq:I_XA_Y_LB}
\end{align}
}\fi
Since (\ref{eq:I_XA_Y_LB}) is valid for any $\tilde{G} \in (0,1]$, 
we choose $\tilde{G}$ to maximize the right hand side of (\ref{eq:I_XA_Y_LB}).
The optimal $\tilde{G}$ is
\begin{align}
\arg\min_{\tilde{G}} \mathbb{E}\left[ (G-\tilde{G})^2 \right] = \mathbb{E}[G].
\end{align}
Since $L = \lceil \sqrt[3]{\beta^2 \SNR} \rceil$ and $\Delta = 1/L$, it follows that
\begin{align}
\lim_{\SNR \rightarrow \infty} {\SNR \Delta} = \infty
\text{ and }
\lim_{\SNR \rightarrow \infty} \SNR \Delta^3 = \frac{1}{\beta^2}.
\end{align}
Moreover, we have (see Appendix \ref{sec:amplitude-fading-moments})
\ifdefined\twocolumnmode{
\begin{align}
\lim_{\SNR \rightarrow \infty} \SNR ~ \textsf{Var}(G) &
= \lim_{\SNR \rightarrow \infty} \SNR \Delta^3 
  \cdot
  \lim_{\Delta \rightarrow 0} \frac{\textsf{Var}(G)}{\Delta^3} \nonumber\\&
= \frac{1}{\beta^2} \cdot \frac{4(\pi\beta)^2}{45}
= \frac{4\pi^2}{45}.
\end{align}
}\else{
\begin{align}
\lim_{\SNR \rightarrow \infty} \SNR ~ \textsf{Var}(G)
= \lim_{\SNR \rightarrow \infty} \SNR \Delta^3 
  \cdot
  \lim_{\Delta \rightarrow 0} \frac{\textsf{Var}(G)}{\Delta^3}
= \frac{1}{\beta^2} \cdot \frac{4(\pi\beta)^2}{45}
= \frac{4\pi^2}{45}.
\end{align}
}\fi
Hence, for $\tilde{G}=\mathbb{E}[G]$, we have
\begin{align}
\lim_{\SNR \rightarrow \infty} I(X_A;Y) - \frac{1}{2} \log{\SNR}
&\geq - 2  - \frac{1}{2} \log(8 \pi)
- \frac{\pi^2}{45}.
\label{eq:I_XA_Y_LB_limit_cubicroot}
\end{align}

\begin{remark}
\label{remark:amplitude-intuition}
We motivate the auxiliary channel in (\ref{eq:Q_V|XA}).
We first express the auxiliary channel as $\tilde{V}$:
\begin{align}
\tilde{V}
&= X_A^2 \Delta \tilde{G} + 2 X_A \Delta \tilde{Z}_1 + \tilde{Z}_0
\label{eq:V_auxiliary}
\end{align}
where $\tilde{G}$, $\tilde{Z}_1$ and $\tilde{Z}_0$ are defined as
\begin{align}
\tilde{G} &\equiv \mathbb{E}[G] 
\label{eq:Gtilde_def} \\
\tilde{Z}_1 &\equiv \sum_{\ell=1}^L \Re[ e^{j \Phi_X} e^{j \Theta_\ell} N_\ell^* ]
\label{eq:Z1tilde_def} \\
\tilde{Z}_0 &\equiv \sigma_N^2.
\label{eq:Z0tilde_def} 
\end{align}
We make this choice for $\tilde{V}$ because
$Z_1$ and $Z_0$ converge (in mean square) to $\tilde{Z}_1$ and $\tilde{Z}_0$, respectively, as $L \rightarrow \infty$.
Moreover, $\tilde{Z}_1$ is Gaussian and $\tilde{Z}_0$ is a constant, 
which are easier quantities to handle than $Z_1$ and $Z_0$.
\end{remark}

\section{Phase Modulation}
\label{sec:approxmodel-phase}
This Appendix provides a proof of Theorem \ref{theorem:oversamp-approxmodel}.
We assume that $\Ts=1$ and $\sigma_N^2=1$ for simplicity.
By using the chain rule, we have
\begin{align}
I(\Phi_{X}^n;\Psiv^n|X_{A}^n)
&= \sum_{k=1}^n I(\Phi_{X,k};\Psiv^n|X_{A}^n,\Phi_{X}^{k-1}) \nonumber \\
&\stackrel{(a)}{\geq} \sum_{k=2}^n I(\Phi_{X,k};\Psiv^n|X_{A}^n,\Phi_{X}^{k-1}) \nonumber \\
&\stackrel{(b)}{\geq} \sum_{k=2}^n I(\Phi_{X,k};\tilde{\Psi}_k |X_{A}^n,\Phi_{X}^{k-1})
\label{eq:I_PhiX1n_Y1n|XA1n_LB1}
\end{align}
where
$\tilde{\Psi}_k$ is a deterministic function of $(\Psiv^n,X_{A}^n,\Phi_{X}^{k-1})$.
Inequality $(a)$ follows from the non-negativity of mutual information
and $(b)$ follows from the data processing inequality.
We choose
\begin{align}
\tilde{\Psi}_k
= \frac{\Psi_{(k-1)L+1}}{\sqrt{\Delta}} \left( \frac{\Psi_{(k-1)L}}{X_{k-1} \Delta}  \right)^*.
\label{eq:Yk_tilde_def}
\end{align}
The intuition behind the choice of (\ref{eq:Yk_tilde_def}) is outlined in Remark \ref{remark:phase-intuition} below.
We express $\tilde{\Psi}_k$ as
\begin{align}
\tilde{\Psi}_k
= \bigg( |X_k| \sqrt{\Delta} e^{j\Phi_{X,k}} + \tilde{N}_k \bigg) \left( 1 + \tilde{Z}_{k-1} \right) e^{j W_{(k-1)L+1}}
\label{eq:Yk_tilde}
\end{align}
where
\begin{align}
\tilde{Z}_k \equiv \frac{N_{kL}^* \ e^{-j \Theta_{kL}} }{X_k^* \Delta}
\label{eq:Zk_tilde_def}
\end{align}
and
\begin{align}
\tilde{N}_{k} \equiv \frac{N_{(k-1)L+1} \ e^{-j \Theta_{(k-1)L+1}} }{\sqrt{\Delta}}.
\label{eq:Nk_tilde_def}
\end{align}
The processes $\tilde{N}_k$ and $\tilde{Z}_{k-1}$ are statistically independent and
\begin{align}
&\tilde{N}_k \sim \mathcal{N}_{\mathbb{C}}(0,1) \\
&\tilde{Z}_{k-1} \Big| \{|X_{k-1}|=|x_{k-1}|\} \sim \mathcal{N}_{\mathbb{C}}\left(0,\frac{1}{|x_{k-1}|^2 \Delta}\right).
\label{eq:Z_given_X}
\end{align}
The notation (\ref{eq:Z_given_X}) means that,
conditioned on $\{|X_{k-1}|=|x_{k-1}|\}$, $\tilde{Z}_{k-1}$ is a Gaussian random variable 
with mean $0$ and variance $1/(|x_{k-1}|^2 \Delta)$.
Moreover, $W_{(k-1)L+1}$ is statistically independent of $\tilde{N}_k$ and $\tilde{Z}_{k-1}$.
The choice of $\tilde{\Psi}_k$ in (\ref{eq:Yk_tilde_def}) implies that
\begin{align}
I(\Phi_{X,k};\tilde{\Psi}_k |X_{A}^n,X^{k-1})
&= I(\Phi_{X,k};\tilde{\Psi}_k |X_{A,k},X_{k-1}).
\label{eq:I_PhiX_Ytilde_LEBOWSKI}
\end{align}
Define $\tilde{\Phi}_{\Psi,k} \equiv \arg({\tilde{\Psi}_k})$ and
\begin{align}
&q_{\tilde{\Phi}_{\Psi}|\Phi_{X}} \left(\phi_y \big| \phi_x \right) 
\equiv \frac{\exp(\alpha \cos(\phi_y-\phi_x) )}{2 \pi I_0(\alpha)}
\label{eq:q_PhiY|PhiX_new}
\end{align}
where $I_0(\cdot)$ is the zeroth-order modified Bessel function of the first kind and $\alpha > 0$.
This distribution is known as Tikhonov (or von Mises) distribution \cite{Mardia2000}.
Furthermore, define
\begin{align}
&q_{\tilde{\Phi}_{\Psi,k}|X_{A,k},X_{k-1}}\left(\phi_y \big| |x_k|,x_{k-1} \right) \nonumber\\&
\equiv \int_{-\pi}^{\pi}
p_{\Phi_{X,k}|X_{A,k},X_{k-1}} \left(\phi_x \big| |x_k|,x_{k-1} \right) 
q_{\tilde{\Phi}_{\Psi}|\Phi_{X}} (\phi_y | \phi_x)
d\phi_x \nonumber\\&
= \frac{1}{2 \pi}.
\label{eq:q_PhiY_new}
\end{align}
The last equality holds because $X_1,\ldots,X_n$ are statistically independent
and $\Phi_{X,k}$ is independent of $X_{A,k}$ with a uniform distribution on $[-\pi,\pi)$.
We have
\ifdefined\twocolumnmode{
\begin{align}
&I(\Phi_{X,k};\tilde{\Psi}_k |X_{A,k},X_{k-1}) \nonumber\\
&\stackrel{(a)}{\geq} I(\Phi_{X,k};\tilde{\Phi}_{\Psi,k} |X_{A,k},X_{k-1}) \nonumber\\
&\stackrel{(b)}{\geq} 
\mathbb{E}\left[\log q_{\tilde{\Phi}_{\Psi}|\Phi_{X}} (\tilde{\Phi}_{\Psi,k}|\Phi_{X,k}) \right] \nonumber\\& - 
\mathbb{E}\left[\log q_{\tilde{\Phi}_{\Psi,k}|X_{A,k},X_{k-1}} \Big(\tilde{\Phi}_{\Psi,k} \big| |X_k|,X_{k-1} \Big) \right] \nonumber\\
&\stackrel{(c)}{=} \log(2\pi) - \log( 2 \pi I_0(\alpha) ) + \alpha \mathbb{E}\left[ \cos(\tilde{\Phi}_{\Psi,k}-\Phi_{X,k}) \right] \nonumber\\
&= - \log( I_0(\alpha) ) + \alpha \mathbb{E}\left[ \cos(\tilde{\Phi}_{\Psi,k}-\Phi_{X,k}) \right] \nonumber\\
&\stackrel{(d)}{\geq} \frac{1}{2} \log\alpha - \alpha 
    + \alpha \mathbb{E}\left[ \cos(\tilde{\Phi}_{\Psi,k}-\Phi_{X,k}) \right] \nonumber\\
&\geq \frac{1}{2} \log{\alpha} - \alpha \frac{\sigma^2_{W}}{2} - \frac{4 \alpha}{\SNR \Delta}
\label{eq:I_PhiX_Ytilde_LB_LEBOWSKI}
\end{align}
}\else{
\begin{align}
&I(\Phi_{X,k};\tilde{\Psi}_k |X_{A,k},X_{k-1}) \nonumber\\
&\stackrel{(a)}{\geq} I(\Phi_{X,k};\tilde{\Phi}_{\Psi,k} |X_{A,k},X_{k-1}) \nonumber\\
&\stackrel{(b)}{\geq} 
\mathbb{E}\left[\log q_{\tilde{\Phi}_{\Psi}|\Phi_{X}} (\tilde{\Phi}_{\Psi,k}|\Phi_{X,k}) \right] - 
\mathbb{E}\left[\log q_{\tilde{\Phi}_{\Psi,k}|X_{A,k},X_{k-1}} \Big(\tilde{\Phi}_{\Psi,k} \big| |X_k|,X_{k-1} \Big) \right] \nonumber\\
&\stackrel{(c)}{=} \log(2\pi) - \log( 2 \pi I_0(\alpha) ) + \alpha \mathbb{E}\left[ \cos(\tilde{\Phi}_{\Psi,k}-\Phi_{X,k}) \right] \nonumber\\
&= - \log( I_0(\alpha) ) + \alpha \mathbb{E}\left[ \cos(\tilde{\Phi}_{\Psi,k}-\Phi_{X,k}) \right] \nonumber\\
&\stackrel{(d)}{\geq} \frac{1}{2} \log\alpha - \alpha 
    + \alpha \mathbb{E}\left[ \cos(\tilde{\Phi}_{\Psi,k}-\Phi_{X,k}) \right] \nonumber\\
&\geq \frac{1}{2} \log{\alpha} - \alpha \frac{\sigma^2_{W}}{2} - \frac{4 \alpha}{\SNR \Delta}
\label{eq:I_PhiX_Ytilde_LB_LEBOWSKI}
\end{align}
}\fi
where
$(a)$ follows from the data processing inequality,
$(b)$ follows by extending the result of the auxiliary-channel lower bound (cf. (\ref{eq:I_AuxCh})),
$(c)$ follows from (\ref{eq:q_PhiY|PhiX_new}) and (\ref{eq:q_PhiY_new}),
$(d)$ follows from
\cite[Lemma 2]{GhozlanISIT2010}
\begin{align}
I_0(z) \leq \frac{\sqrt{\pi}}{2} \frac{e^z}{\sqrt{z}} \leq \frac{e^z}{\sqrt{z}}
\label{eq:BesselI0_UB}
\end{align}
and the last inequality holds because
\begin{align}
\mathbb{E}\left[\cos(\tilde{\Phi}_{\Psi,k}-\Phi_{X,k})\right] 
\geq 1 - \frac{\sigma^2_{W}}{2} - \frac{4}{\SNR \Delta}
\end{align}
for $\SNR \Delta > 2$, 
as we will show.
Define
\begin{align}
S_k 
&\equiv 1 + \tilde{Z}_k
\label{eq:Sk_def} \\
\hat{\Psi}_k
&\equiv |X_k| \sqrt{\Delta} e^{j\Phi_{X,k}} + \tilde{N}_k.
\label{eq:Yk_hat_def}
\end{align}
Therefore, we have
\begin{align}
&  \mathbb{E}\left[\cos(\tilde{\Phi}_{\Psi,k}-\Phi_{X,k})\right] \nonumber\\
&\stackrel{(a)}{=} \mathbb{E}\left[\cos(W_{(k-1)L+1}+\Phi_{S,k-1}+\hat{\Phi}_{\Psi,k}-\Phi_{X,k}) \right] \nonumber\\
&\stackrel{(b)}{=} \mathbb{E}\left[\cos(W_{(k-1)L+1}+\Phi_{S,k-1})\right] \mathbb{E}\left[\cos(\hat{\Phi}_{\Psi,k}-\Phi_{X,k}) \right] \nonumber\\&
 - \mathbb{E}\left[\sin(W_{(k-1)L+1}+\Phi_{S,k-1})\right] \mathbb{E}\left[\sin(\hat{\Phi}_{\Psi,k}-\Phi_{X,k}) \right] 
\label{eq:Ecos_PhiYtilde_minus_PhiX}
\end{align}
where $\Phi_{S,k} \equiv \arg({S_k})$ 
and $\hat{\Phi}_{\Psi,k} \equiv \arg({\hat{\Psi}_k})$. 
Step
$(a)$ follows from (\ref{eq:Yk_tilde})
while step
$(b)$ follows from the trigonometric relation 
$\cos(A+B) = \cos(A) \cos(B) - \sin(A) \sin(B)$
and because $W_{(k-1)L+1}+\Phi_{S,k-1}$ is independent of $\hat{\Phi}_{\Psi,k}$ and $\Phi_{X,k}$.

By using Lemma \ref{lemma:awgn-lemma-trig} in Appendix \ref{sec:awgn-phase}, we have
\begin{align}
\mathbb{E}\left[ \cos(\hat{\Phi}_{\Psi,k}-\Phi_{X,k}) \bigg| X_{A,k}=|x_k| \right] 
&\geq 1 - \frac{1}{|x_k|^2 \Delta} 
\label{eq:Ecos_PhiYhat_minus_PhiX} \\
\mathbb{E}\left[ \sin(\hat{\Phi}_{\Psi,k}-\Phi_{X,k}) \bigg| X_{A,k}=|x_k| \right]
&=0.
\label{eq:Esin_PhiYhat_minus_PhiX}
\end{align}
Equations (\ref{eq:Ecos_PhiYtilde_minus_PhiX}) and (\ref{eq:Esin_PhiYhat_minus_PhiX}) imply that 
we do not need to compute $\mathbb{E}\left[\sin(W_{(k-1)L+1}+\Phi_{S,k-1})\right]$.
We have
\ifdefined\twocolumnmode{
\begin{align}
&  \mathbb{E}\left[\cos(W_{(k-1)L+1}+\Phi_{S,k-1})\right] \nonumber\\
&= \mathbb{E}\left[\cos(W_{(k-1)L+1})\right] \mathbb{E}\left[\cos(\Phi_{S,k-1})\right] \nonumber\\&
 - \mathbb{E}\left[\sin(W_{(k-1)L+1})\right] \mathbb{E}\left[\sin(\Phi_{S,k-1})\right]
\label{eq:Ecos_W_plus_PhiS}
\end{align}
}\else{
\begin{align}
&  \mathbb{E}\left[\cos(W_{(k-1)L+1}+\Phi_{S,k-1})\right] \nonumber\\
&= \mathbb{E}\left[\cos(W_{(k-1)L+1})\right] \mathbb{E}\left[\cos(\Phi_{S,k-1})\right] 
 - \mathbb{E}\left[\sin(W_{(k-1)L+1})\right] \mathbb{E}\left[\sin(\Phi_{S,k-1})\right]
\label{eq:Ecos_W_plus_PhiS}
\end{align}
}\fi
because $W_{(k-1)L+1}$ and $\Phi_{S,k-1}$ are independent.
Since $W_{(k-1)L+1}$ is Gaussian with mean 0 and variance $\sigma^2_W = 2\pi\beta \Delta$,
the characteristic function of $W_{(k-1)L+1}$ is
\begin{align}
\mathbb{E}\left[ e^{j W_{(k-1)L+1}} \right]
= e^{- \sigma^2_{W}/2}
\label{eq:E_expjW}
\end{align}
and by using the linearity of expectation, we have
\begin{align}
&\mathbb{E}\left[ \cos(W_{(k-1)L+1}) \right]
= \Re\left[ \mathbb{E}\left[ e^{j W_{(k-1)L+1}} \right] \right]
= e^{- \sigma^2_{W}/2} 
\label{eq:Ecos_W} \\
&\mathbb{E}\left[ \sin(W_{(k-1)L+1}) \right]
= \Im\left[ \mathbb{E}\left[ e^{j W_{(k-1)L+1}} \right] \right]
= 0
\label{eq:Esin_W}
\end{align}
where $\Re[\cdot]$ and $\Im[\cdot]$ denote the real and imaginary parts, respectively.
By using Lemma \ref{lemma:awgn-lemma-trig} in Appendix \ref{sec:awgn-phase} again, we have
\begin{align}
\mathbb{E}\left[\cos(\Phi_{S,k-1}) \bigg| X_{k-1}=x_{k-1} \right]
&\geq 1 - \frac{1}{|x_{k-1}|^2 \Delta}.
\label{eq:Ecos_PhiS}
\end{align}
It follows from (\ref{eq:Ecos_W_plus_PhiS}), (\ref{eq:Ecos_W}), (\ref{eq:Esin_W}) and (\ref{eq:Ecos_PhiS}) that
\ifdefined\twocolumnmode{
\begin{align}
&\mathbb{E}\left[\cos(W_{(k-1)L+1}+\Phi_{S,k-1})\right] \nonumber\\&
\geq e^{- \sigma^2_{W}/2} \left(1 - \mathbb{E}\left[\frac{1}{|X_{k-1}|^2}\right] \frac{1}{\Delta} \right).
\label{eq:Ecos_W_plus_PhiS_LB}
\end{align}
}\else{
\begin{align}
\mathbb{E}\left[\cos(W_{(k-1)L+1}+\Phi_{S,k-1})\right] 
\geq e^{- \sigma^2_{W}/2} \left(1 - \mathbb{E}\left[\frac{1}{|X_{k-1}|^2}\right] \frac{1}{\Delta} \right).
\label{eq:Ecos_W_plus_PhiS_LB}
\end{align}
}\fi
By combining (\ref{eq:Ecos_PhiYtilde_minus_PhiX}), (\ref{eq:Ecos_PhiYhat_minus_PhiX}), (\ref{eq:Esin_PhiYhat_minus_PhiX})
(\ref{eq:Ecos_W_plus_PhiS_LB}), we have
(for $P\Delta > 2$)
\begin{align}
&\mathbb{E}\left[\cos(\tilde{\Phi}_{\Psi,k}-\Phi_{X,k})\right] \nonumber\\
&\geq e^{- \sigma^2_{W}/2} 
      \left(1 - \mathbb{E}\left[\frac{1}{|X_{k-1}|^2}\right] \frac{1}{\Delta} \right) 
			\left(1 - \mathbb{E}\left[\frac{1}{|X_k|^2    }\right] \frac{1}{\Delta} \right) \nonumber \\
&\stackrel{(a)}{\geq} e^{- \sigma^2_{W}/2} \left(1 - \frac{2}{P \Delta}\right)^2 \nonumber \\
&\stackrel{(b)}{\geq} e^{- \sigma^2_{W}/2} - \frac{4}{P \Delta} \nonumber \\
&\stackrel{(c)}{\geq} 1 - \frac{\sigma^2_{W}}{2} - \frac{4}{P \Delta}
\label{eq:Ecos_PhiYtilde_minus_PhiX_LB}
\end{align}
where 
$(a)$ holds because $|X_k|^2 \geq P/2$ and $|X_{k-1}|^2 \geq P/2$,
while
$(b)$ follows from $e^{- \sigma^2_{W}/2} \leq 1$ and the non-negativity of $4/(P\Delta)^2$ and
$(c)$ follows from $e^{-x} \geq 1 - x$.

It follows from (\ref{eq:I_PhiX1n_Y1n|XA1n_LB1}), (\ref{eq:I_PhiX_Ytilde_LEBOWSKI}) and (\ref{eq:I_PhiX_Ytilde_LB_LEBOWSKI}) that
\begin{align}
\frac{1}{\nsymb} I(\Phi_{X}^\nsymb;\Psiv^\nsymb|X_{A}^\nsymb)
&\geq \frac{\nsymb-1}{\nsymb}
\left[ \frac{1}{2} \log{\alpha} - \alpha \pi\beta\Delta - \frac{4 \alpha}{\SNR \Delta} \right].
\label{eq:I_PhiX1n_Y1n|XA1n_LB2}
\end{align}
Hence, we have
\begin{align}
I(\Phi_{X};\Psi|X_{A})
&= \lim_{\nsymb \rightarrow \infty} \frac{1}{\nsymb} I(\Phi_{X}^\nsymb;\Psiv^\nsymb|X_{A}^\nsymb) \nonumber \\
&\geq \frac{1}{2} \log{\alpha} - \alpha \pi\beta\Delta - \frac{4 \alpha}{\SNR \Delta}.
\label{eq:I_PhiX_Y|XA_LB_LEBOWSKI}
\end{align}
Since the lower bound (\ref{eq:I_PhiX_Y|XA_LB_LEBOWSKI}) is valid for any $\alpha>0$,
we have
\begin{align}
I(\Phi_{X};\Psi|X_{A})
&\geq \max_{\alpha>0} \frac{1}{2} \log{\alpha} - \alpha \left( \pi \beta \Delta + \frac{4}{\SNR \Delta} \right) \nonumber \\
&= - \frac{1}{2} \log\left(2 \pi \beta \Delta + \frac{8}{\SNR \Delta}\right) - \frac{1}{2} \nonumber\\
&= - \frac{1}{2} \log\left(\frac{2 \pi \beta}{L} + \frac{8 L}{\SNR}\right) - \frac{1}{2}
\label{eq:I_PhiX_Y|XA_LB_best_alpha}
\end{align}
where we used $\Delta = 1/L$ in the last equation.
We maximize the lower bound (\ref{eq:I_PhiX_Y|XA_LB_best_alpha}) over $L$ to 
obtain \footnote{Ignoring $L \in \mathbb{Z}$.}  
\begin{align}
L = \sqrt{\pi \beta \SNR}/2.
\end{align}
By choosing an oversampling factor
$L = \lceil \sqrt{\pi \beta \SNR}/2 \rceil$,
it follows that
\begin{align}
\lim_{\SNR \rightarrow \infty} I(\Phi_{X};\Psi|X_{A}) - \frac{1}{4} \log{\SNR}
\geq - \frac{1}{4} \log(64 \pi \beta) - \frac{1}{2}.
\label{eq:inforate-phase-contrib-best_alpha}
\end{align}


The proof of (\ref{eq:I_XA_Y_approxmodel_theorem}) is similar to the proof in 
Appendix \ref{sec:amplitude-lower-bound} for the full model and is omitted.\footnote{
The inequality (\ref{eq:I_XA_Y_LB}) is valid for the approximate model except that $G$ is set to $1$ because $F_k=1$ in the approximate model. 
Hence, by setting $\tilde{G}=1$, the term $-\pi^2/45$ in (\ref{eq:I_XA_Y_fullmodel_theorem}) does not appear in (\ref{eq:I_XA_Y_approxmodel_theorem}).}

\begin{remark}
\label{remark:phase-intuition}
We outline the intuition for choosing $\tilde{\Psi}_k$ in (\ref{eq:Yk_tilde_def}).
The basic idea is to process $(\Psiv^n,X_{A}^n,\Phi_{X}^{k-1})$ in three steps:
first estimate the phase noise, then cancel it, then decode.
We outline the motivation for each step:
\begin{enumerate}[wide, labelwidth=!, labelindent=0pt, itemsep=1pt]
\item 
Since the past inputs $X^{k-1}$ and outputs $\Y^{k-1}$ are available, one may use the simple estimator
\begin{align}
e^{j \widehat{\Theta}_{k^\prime}}
\equiv \frac{\Psi_{k^\prime}}{X_{\lceil k^\prime / L \rceil } \Delta} 
\label{eq:Theta_hat}
\end{align}
to estimate $\Theta_{k^\prime}$ for $k^\prime = 1,\ldots,(k-1)L$.
This is a reasonable estimate as long as the energy per sample is large relative to the noise variance per sample,
i.e., when $\SNR \Delta \gg 1$.

\item
Next, we use the previous estimates to cancel from $\Y_k^n$ the phase noise that accumulated up to sample $(k-1)L$.
More precisely, we use only the most recent estimate $\widehat{\Theta}_{(k-1)L}$.
We make this choice because
the discrete-time Wiener process $\{\Theta_k\}$ is a first-order Markov process
so, conditioned on $\Theta_{kL}$, the random variables $\Theta_{kL+1}$,...,$\Theta_{nL}$ are independent of $\Theta_{0}$,...,$\Theta_{kL-1}$.
This means that if the SNR is high enough so that
$\widehat{\Theta}_{k^\prime}$ is an accurate estimate of $\Theta_{k^\prime}$ for $k^\prime = 1$,$\ldots$, $(k-1)L$,
then it would suffice to use only $\widehat{\Theta}_{(k-1)L}$.

\item
Since the input is i.i.d. and there is no inter-symbol interference,
we discard $\Psi_{k+1}^n$ and use only $\Psiv_k$ (after canceling the phase noise) for decoding $\Phi_{X,k}$.
Moreover, we use only the first sample $\Psi_{(k-1)L+1}$ in $\Psiv_k$.
This is because, in absence of additive noise, we have
\begin{align*}
&I(\Phi_{X,k};\Psiv_k) \nonumber\\
&=I(\Phi_{X,k};\{\Phi_{\Psi,(k-1)L+1}:1 \leq \ell \leq L\}) \nonumber\\
&=I(\Phi_{X,k};\Phi_{\Psi,(k-1)L+1},\{W_{(k-1)L+\ell}:2 \leq \ell \leq L\}) \nonumber\\
&=I(\Phi_{X,k};\Phi_{\Psi,(k-1)L+1}) \nonumber\\
&=I(\Phi_{X,k};\Psi_{(k-1)L+1}).
\end{align*}

\end{enumerate}
\end{remark}

\section{Supporting Lemma for Appendix \ref{sec:amplitude-lower-bound}}
\label{sec:amplitude-fading-moments}
The following lemma is used in Appendix \ref{sec:amplitude-lower-bound}.

\begin{lemma} \label{lemma:limit_VarG}
For $G$ in (\ref{eq:G_def}) with $\Ts=1$, we have
\begin{align}
 \lim_{\Delta \rightarrow 0} \frac{\textsf{Var}(G)}{\Delta^3}
 = \frac{4}{45} (\pi \beta)^2.
\label{eq:limits_VarG}
\end{align}
\end{lemma}
\begin{proof}
Using the definition of $G$ in (\ref{eq:G_def}), we have
\begin{align}
\textsf{Var}(G) 
= \textsf{Var}\left( \frac{1}{L} \sum_{\ell=1}^L |F_\ell|^2 \right)
= \frac{1}{L} \textsf{Var}(|F_1|^2)
\label{eq:VarG}
\end{align}
because $\{F_k\}$ is an i.i.d. process.
The independence follows because
increments of a Wiener process over disjoint time intervals are independent
and because functions of independent random variables are also independent.
The identical distribution property follows from the rectangular pulse shape\footnote{
If the pulse is not rectangular, the process $\{F_k\}$ may not be identically distributed 
because the phase noise $e^{j\Theta(t)}$ is weighted in a given sample interval according to the pulse shape.}
and that the integrations are over increments of a Wiener process with equal length.
By using the formula in \cite[Theorem 1]{ShamaiIT94} 
for computing moments of filtered Wiener phase noise 
(specifically $\mathbb{E}[|F_1|^{2k}]$ where $k$ is a positive integer),
we have
\begin{align}
\mathbb{E}[ |F_1|^2 ]
= 2\frac{a-1-\log{a}}{(\log{a})^2}
\label{eq:EF2}
\end{align}
and
\ifdefined\twocolumnmode{
\begin{align}
&\mathbb{E}[ |F_1|^4 ] \nonumber \\&
= \frac{783 - 784 a + a^4 + (540 + 240 a) \log{a} + 144 (\log{a})^2}{18 (\log{a})^4}.
\label{eq:EF4}
\end{align}
}\else{
\begin{align}
\mathbb{E}[ |F_1|^4 ] 
= \frac{783 - 784 a + a^4 + 540 \log{a} + 240 a \log{a} + 144 (\log{a})^2}{18 (\log{a})^4}
\label{eq:EF4}
\end{align}
}\fi
where
\begin{align}
a \equiv e^{-\pi \beta \Delta}.
\label{eq:a_def}
\end{align}
It follows from (\ref{eq:EF4}) and (\ref{eq:EF2}) that
\ifdefined\twocolumnmode{
\begin{align}
&\textsf{Var}(|F_1|^2) \nonumber\\
&= \mathbb{E}[|F_1|^4] - (\mathbb{E}[|F_1|^2])^2  \nonumber\\
&= \frac{711 - 640 a - 72 a^2 + a^4 + 12 (33 + 32 a) \log{a} + 72 (\log{a})^2}{18 (\log{a})^4}
\label{eq:VarF2}
\end{align}
}\else{
\begin{align}
\textsf{Var}(|F_1|^2)
&= \mathbb{E}[|F_1|^4] - (\mathbb{E}[|F_1|^2])^2  \nonumber\\
&= \frac{711 - 640 a - 72 a^2 + a^4 + 396 \log{a} + 384 a \log{a} + 72 (\log{a})^2}{18 (\log{a})^4}
\label{eq:VarF2}
\end{align}
}\fi
which implies that
\begin{align}
 \lim_{a \rightarrow 1} \frac{\textsf{Var}(|F_1|^2)}{(\log{a})^2}
 = \frac{4}{45}.
\label{eq:limits_VarF2}
\end{align}
Finally, (\ref{eq:limits_VarG}) follows from 
(\ref{eq:limits_VarF2}), (\ref{eq:a_def}) and $L=1/\Delta$ (because $\Ts=1$).
\end{proof}

\section{Supporting Lemmas for Appendix \ref{sec:approxmodel-phase}}
\label{sec:awgn-phase}
This appendix contains two lemmas.
The main lemma is Lemma \ref{lemma:awgn-lemma-trig} which is used in Appendix \ref{sec:approxmodel-phase}.
Lemma \ref{lemma:lowerbounds} is a supporting lemma for Lemma \ref{lemma:awgn-lemma-trig}.

\begin{lemma} \label{lemma:awgn-lemma-trig}
Let $R$ be a fixed positive real number and let $\Phi \equiv \arg(Y)$ where
\begin{align}
Y = R + Z
\label{eq:ring_Y_def}
\end{align}
and $Z$ is a circularly-symmetric complex Gaussian random variable with mean $0$ and variance $1$.
Then, we have
\begin{align}
\mathbb{E}\left[ \cos(\Phi) \right] &\geq 1 - \frac{1}{R^2}
\label{eq:Ecos} \\
\mathbb{E}\left[ \sin(\Phi) \right] &= 0.
\label{eq:Esin}
\end{align}
\end{lemma}

\begin{proof}
The probability density function of $\Phi$ is \cite{Aldis1993}
\begin{align}
p_{\Phi}(\phi) 
&= \frac{1}{2 \pi} e^{-\R^2}
+ \frac{1}{\sqrt{4 \pi}} \R \cos(\phi) \ e^{-\R^2 \sin^2\phi}
\erfc\left(- \R \cos\phi \right)
\label{eq:p_Phi}
\end{align}
where $\erfc(\cdot)$ is the complementary error function
\begin{align}
\erfc(z) \equiv \frac{2}{\sqrt{\pi}} \int_z^\infty e^{-t^2} dt.
\label{eq:erfc_def}
\end{align}
Therefore, we have
\begin{align}
&\mathbb{E}\left[ \cos(\Phi) \right] \nonumber\\
&\equiv \int_{-\pi}^{\pi} \cos(\phi) \ p_{\Phi}(\phi) \ d\phi \nonumber\\
&\stackrel{(a)}{=} 2 \int_{0}^{\pi} \cos(\phi) \ p_{\Phi}(\phi) \ d\phi \nonumber\\
&\stackrel{(b)}{=} \int_{0}^{\pi} \cos(\phi) \left[
\frac{1}{\sqrt{\pi}} \R \cos(\phi) \ e^{-\R^2 \sin^2\phi}
\erfc\left(- \R \cos\phi \right)
\right] d\phi \nonumber\\
&\stackrel{(c)}{\geq} \int_{0}^{\pi/2} 
\frac{\R}{\sqrt{\pi}} \cos^2(\phi) \ e^{-\R^2 \sin^2\phi}
\erfc\left(- \R \cos\phi \right) d\phi \nonumber\\
&\stackrel{(d)}{\geq} \int_{0}^{\pi/2} 
\frac{\R}{\sqrt{\pi}} \cos^2(\phi) \ e^{-\R^2 \sin^2\phi}
\left(2 - e^{-\R^2 \cos^2\phi} \right) d\phi \nonumber\\
&= \int_{0}^{\pi/2} \frac{2 \R}{\sqrt{\pi}} \cos^2(\phi) \ e^{-\R^2 \sin^2\phi} d\phi
- \frac{\R}{\sqrt{\pi}} e^{-\R^2} \int_{0}^{\pi/2} \cos^2\phi d\phi \nonumber\\
&\stackrel{(e)}{=} \int_{0}^{\pi/2} \frac{2 \R}{\sqrt{\pi}} \cos^2(\phi) \ e^{-\R^2 \sin^2\phi} d\phi
 - \frac{\sqrt{\pi}}{4} \ \R e^{-\R^2}
\label{eq:Ecos_Phi_LB1}
\end{align}
where
$(a)$ follows because both $p_{\Phi}(\cdot)$ and $\cos(\phi)$ are even functions,
$(b)$ follows from (\ref{eq:p_Phi}) and $\int_0^\pi \cos\phi ~ d\phi = 0$,
$(c)$ follows because the integrand is non-negative over the interval $\phi \in (\pi/2,\pi]$,
$(d)$ holds because $\erfc\left(- \R \cos\phi \right) \geq 2 - e^{-\R^2 \cos^2\phi}$ 
(see Lemma \ref{lemma:lowerbounds}) 
and $\cos^2(\phi) \ e^{-\R^2 \sin^2\phi} \geq 0$
and finally 
$(e)$ follows by direct integration.

We bound the integral in (\ref{eq:Ecos_Phi_LB1}) as follows
\begin{align}
&\int_{0}^{\pi/2} \frac{2 \R}{\sqrt{\pi}} \cos^2\phi \ e^{-\R^2 \sin^2\phi} d\phi \nonumber\\
&\stackrel{(a)}{\geq} \frac{2 \R}{\sqrt{\pi}} \int_{0}^{\pi/2} (1-\phi^2) \ e^{-\R^2 \phi^2} d\phi \nonumber\\
&= \frac{2 \R}{\sqrt{\pi}} \int_{0}^{\pi/2} e^{-\R^2 \phi^2} d\phi 
 - \frac{2 \R}{\sqrt{\pi}} \int_{0}^{\pi/2} \phi^2 \ e^{-\R^2 \phi^2} d\phi \nonumber\\
&\stackrel{(b)}{=} \erf\left( \frac{\pi}{2} R \right)
- \frac{1}{2 R^2} \erf\left( \frac{\pi}{2} R \right) + \frac{1}{2\sqrt{\pi} R} e^{-\pi^2 R^2/4} \nonumber\\
&\stackrel{(c)}{\geq} \erf\left( \frac{\pi}{2} R \right)
- \frac{1}{2 R^2}
\end{align}
where $\erf(\cdot)$ is the error function defined as
\begin{align}
\erf(z) \equiv \frac{2}{\sqrt{\pi}} \int_0^z e^{-t^2} dt.
\label{eq:erf_def}
\end{align}
Step 
$(a)$ follows from $\cos^2{\phi} \ e^{-R^2 \sin^2{\phi}} \geq (1-\phi^2) e^{-R^2 \phi^2}$ 
(see Lemma \ref{lemma:lowerbounds}),
$(b)$ follows by using the definition of $\erf(\cdot)$ and \cite[3.321(5)]{Gradshtein2007}:
\begin{align}
\int_0^b t^2 e^{-a^2 t^2} dt
&= \frac{1}{2 a^3} \left[ \frac{\sqrt{\pi}}{2} \erf\left( a b \right) - a b \ e^{-a^2 b^2} \right]
\end{align}
and 
$(c)$ holds because $\erf(\cdot) \leq 1$ and the last term is non-negative.
Substituting back in (\ref{eq:Ecos_Phi_LB1}) yields
\begin{align}
\mathbb{E}\left[ \cos(\Phi) \right] 
&\geq \erf\left( \frac{\pi}{2} R \right)
- \frac{1}{2 R^2}
- \frac{\sqrt{\pi}}{4} \R \ e^{-\R^2} \\
&\geq 1 - e^{-\pi^2 R^2/4}
- \frac{1}{2 R^2}
- \frac{\sqrt{\pi}}{4} \R \ e^{-\R^2}
\label{eq:Ecos_Phi_LB}
\end{align}
where the last inequality
follows from the relations $\erf\left( {\pi R}/{2} \right) =  1 - \erfc\left( {\pi R}/{2} \right)$ 
and $\erfc\left( {\pi R}/{2} \right) \leq e^{-\pi^2 R^2/4}$ (see \cite[eq. (5)]{Chiani2002}).
Therefore, we have
\begin{align}
R^2 \mathbb{E}\left[ \cos(\Phi) \right]
&\geq R^2 - \frac{4}{\pi^2 e} - \frac{1}{2}
 - \frac{\sqrt{\pi}}{4} \left(\frac{3}{2 e}\right)^{3/2}
\label{eq:R2_Ecos_Phi_LB}
\end{align}
because 
\begin{align}
 \max_{x} x^m e^{-a x^2}  = \left( \frac{m}{2 a e} \right)^{m/2}
\label{eq:max_power_exp}
\end{align}
for any $a > 0 $ and positive integer $m$,
which implies that
$\R^3 e^{-\R^2} \leq \left({3}/{2 e}\right)^{3/2}$ and
$\R^2 e^{-\pi^2 \R^2/4} \leq {4}/{\pi^2 e}$.
We conclude that
\begin{align}
\mathbb{E}\left[ \cos(\Phi) \right]
\geq 1 - \frac{0.83}{R^2}
\geq 1 - \frac{1}{R^2}.
\label{eq:Ecos_Phi_LB2}
\end{align}

Finally, we have
\begin{align}
\mathbb{E}\left[ \sin(\Phi) \right] 
\equiv \int_{-\pi}^{\pi} \sin(\phi) \ p_{\Phi}(\phi) \ d\phi
= 0
\label{eq:Esin_Phi}
\end{align}
because the integrand is odd.

\end{proof}

\begin{lemma} \label{lemma:lowerbounds}
The following bounds hold.
\begin{enumerate}
\item 
For $z \geq 0$:
\begin{align}
\erfc(-z) \geq 2-e^{-z^2}.
\end{align}

\item 
For $a \geq 0$ and $ t \geq 0$:
\begin{align}
\cos^2{t} \ e^{-a \sin^2{t}} \geq (1-t^2) e^{-a t^2}.
\end{align}
\end{enumerate}
\end{lemma}

\begin{proof}
\begin{enumerate}
\item 
Since for any real number $z$
\begin{align}
\erf(-z) &= - \erf(z) \\
\erfc(z) &= 1-\erf(z)
\end{align}
we have
\begin{align}
\erfc(-z) = 1-\erf(-z) = 1+\erf(z) = 2-\erfc(z).
\end{align}
To complete the proof, we use \cite[eq. (5)]{Chiani2002}
\begin{align}
\erfc(z) \leq e^{-z^2} \ \text{for} \ z \geq 0.
\end{align}

\item	
The trigonometric identity $\cos^2{t} + \sin^2{t} = 1 $
and the bound $\sin{t} \leq t$ (for $t \geq 0$) imply that
\begin{align}
\cos^2{t} \geq 1-t^2.
\end{align} 

Since the exponential function is monotone,
we have
\begin{align}
e^{-a \sin^2{t}} \geq e^{-a t^2}.
\end{align}

By combining the two inequalities, we obtain
\begin{align}
\cos^2{t} \ e^{-a \sin^2{t}} 
\geq (1-t^2) e^{-a t^2}.
\end{align}
\end{enumerate}
\end{proof}




%

\appendices


\ifCLASSOPTIONcaptionsoff
  \newpage
\fi



%
\bibliographystyle{IEEEtran}
\bibliography{bibstrings_short,ieee_dissert4_short}

%

%





\end{document}